\documentclass[10pt]{article} % For LaTeX2e
\usepackage[preprint]{tmlr}

\usepackage{algorithm}
\usepackage{multirow}
\usepackage{algorithmic}
\usepackage{natbib}
\usepackage{microtype}
\usepackage{booktabs} % for professional tables
\usepackage{hyperref}

\usepackage{mathtools}
\usepackage{amsmath, amsfonts, amsthm, graphicx, amssymb, mathtools, nicefrac, bbm}

\usepackage[capitalize,noabbrev]{cleveref}

\usepackage{breakurl}
\newtheorem{theorem}{Theorem}[section]

\newtheorem{lemma}[theorem]{Lemma}

\newtheorem{definition}[theorem]{Definition}
\newtheorem{assumption}[theorem]{Assumption}
\newtheorem{remark}[theorem]{Remark}
 
\newtheorem{claim}[theorem]{Claim}

\usepackage{graphicx}
\usepackage{physics}
\usepackage{esvect}
\usepackage{cancel}

\newcommand{\E}{\mathbb{E}}
\newcommand{\Rorth}{\mathbb{R}_{\geq 0}}
\newcommand{\Rpos}{\mathbb{R}_{>0}}
\newcommand{\cD}{\mathcal{D}}
\newcommand{\cS}{\mathcal{S}}

\newcommand\C{\mathcal{C}}

\newcommand{\defn}{:=}

\newcommand{\numprod}{n}

\newcommand{\sminus}{s_{-i}}

\newcommand{\expnt}[1]{\exp \left( #1 \right)}
\def\code#1{\texttt{#1}} % for code 

\usepackage{xspace}
\newcommand{\rs}{recommender system\xspace}
\usepackage{relsize}

\usepackage{todonotes}
\usepackage{xcolor}

\DeclareMathOperator*{\argmax}{arg\,max}

\def\code#1{\texttt{#1}} % for code 

\usepackage{xfrac}
\usepackage[utf8]{inputenc}
\usepackage{hyperref}
\usepackage{xcolor}
\usepackage{url}
\usepackage{bm}
\usepackage{mathtools}
\usepackage{authblk}
\usepackage{mdframed}
\usepackage{lipsum}

\usepackage{wrapfig}
\usepackage{graphicx}
\graphicspath{ {./Figures/} }
\usepackage{subfloat}
\usepackage{caption}
\usepackage{subcaption}

\usepackage{hyperref}
\usepackage{url}

\title{Producers Equilibria and Dynamics in Engagement-Driven Recommender Systems}

\author{\name Krishna Acharya \email krishna.acharya@gatech.edu \\
    \addr Georgia Institute of Technology
    \AND
    \name Varun Vangala \email vvangala3@gatech.edu \\
    \addr Georgia Institute of Technology
    \AND
    \name Jingyan Wang \email jingyanw@ttic.edu \\
    \addr Toyota Technological Institute at Chicago
    \AND
    \name Juba Ziani \email jziani3@gatech.edu \\
    \addr Georgia Institute of Technology
}
\def\month{02}  % Insert correct month for camera-ready version
\def\year{2025} % Insert correct year for camera-ready version
\def\openreview{\url{https://openreview.net/forum?id=EWT4GxjGDS}} % Insert correct link to OpenReview for camera-ready version

\let\classAND\AND
\let\AND\relax
\usepackage{algorithmic}
\let\AND\classAND
\AtBeginEnvironment{algorithmic}{\let\AND\algoAND}

\floatname{algorithm}{Algorithm}

\begin{document}

\maketitle

\begin{abstract}
Online platforms such as YouTube, Instagram heavily rely on recommender systems to decide what content to present to users. Producers, in turn, often create content that is likely to be recommended to users and have users engage with it. To do so, producers try to align their content with the preferences of their targeted user base. In this work, we explore the equilibrium behavior of producers who are interested in maximizing user \emph{engagement}. We study two variants of the content-serving rule for the platform's recommender system, and provide a structural characterization of producer behavior at equilibrium: namely, each producer chooses to focus on a single embedded feature.
We further show that specialization, defined as different producers optimizing for distinct types of content, naturally emerges from the competition among producers trying to maximize user engagement. We provide a heuristic for computing equilibria of our engagement game, and evaluate it experimentally. We highlight i) the performance and convergence of our heuristic, ii) the degree of producer specialization, and iii) the impact of the content-serving rule on producer and user utilities at equilibrium and provide guidance on how to set the content-serving rule \footnote{Code available at \url{https://github.com/krishnacharya/recsys_eq}}.
\end{abstract}

\section{Introduction}
Recommender systems have transformed our interactions with online platforms like Instagram, Spotify, and YouTube \citep{stiger19, qian2022digital}. These systems curate content to enhance user experiences, fostering user retention and engagement \citep{Goodrow_2021, Instagram_2023}. This often translates into increased revenue for the platform. Due to their importance to today's digital industry, there has been a large body of work aiming at developing new and improving existing recommendation algorithms, e.g.,~\citep{koren2009matrix,li2010contextual,lu2012recommender,wang2012nonnegative,luo2014efficient,covington2016deep, he2017neural, two-tower2}.

In 2023, the creator economy, driven largely by these systems, was valued at a staggering \$250 billion~\citep{businessinsider}. Content producers have adapted to this landscape to act \emph{strategically} \citep{milli2023choosing, merrill2021five, Mack_2019}. They often compete against each other and tailor their content to maximize \emph{exposure} (or how many users does a producer reach)  or \emph{engagement} (or how much users engage with a given producer's content). This competition can be modelled as a game, where understanding the dynamics offers insights into content creation incentives and phenomena like content specialization. However, comprehending these dynamics, both theoretically and empirically, presents challenges when considering the complex nature of user preferences and content strategies in high-dimensional spaces. Finding Nash equilibria is in general a computationally difficult process, especially as the number of players and the action space increase. In the case of recommender systems, user preferences and producer contents are represented by high-dimensional vectors, and the set of possible strategy profiles for the producers rapidly grows intractably large as the size of the game increases. 

In this work, we aim to provide new insights into producer competition and the equilibria of recommender systems. We take a departure from much of the related literature that aims at understanding equilibria and dynamics when producers try to maximize \emph{exposure} (i.e., how many users see their content). We instead focus on producers interested in maximizing \emph{engagement}, a metric that encompasses not just exposure but also how much users interact with content. 

Our main goal is to understand the extent and conditions under which content specialization occurs—defined as different producers choosing to create distinct types of content—instead of all producers producing the same homogeneous content.

\paragraph{Summary of contributions} Our main contributions are as follows:
\begin{itemize}
\item In Section~\ref{sec:model}, we formally introduce our model. We assume producers aim to maximize user \emph{engagement} instead of user \emph{exposure}, where the latter is typical in works characterizing producer equilibria in recommender systems \citep{hron2022modeling, meena-ss}. We rely on the softmax rule used in previous work~\citep{hron2022modeling, chen2019top} for showing content to users, but also introduce a new linear-proportional content-serving rule as a baseline for comparison.
\item In Section~\ref{sec:eq_structure}, we provide our main theoretical structural characterization result. We \emph{mathematically prove} that at equilibrium, producers prefer producing content that targets a single (embedded) feature at a time, rather than investing across several features.
\item In Section~\ref{sec:simple_eq}, we characterize the pure Nash equilibrium with the linear content-serving rule for a simplified setting in which we assume all users are ``single-minded'', i.e. are only interested in a single type of content. The closed-form equilibrium we derive exhibits \emph{specialization}, defined as when different producers split themselves across different types of content, rather than all producing the same, homogeneous content. While this setting is simplified, we note in the experimental results of Section~\ref{sec:experiments} that the insights it provides do carry on to the general softmax content-serving rule and to more general user preferences.
\item In Section~\ref{sec:experiments}, we present experiments on synthetic data and three real-world datasets—Movielens-100k, Amazon Music, and Rent the Runway~\citep{movielens,amazon-ratings,rentrunway}. First, we introduce a computationally efficient heuristic based on best-response dynamics (Algorithm~\ref{alg:bestrep_dynamics}) for computing pure Nash equilibria of our engagement game. We observe that this heuristic almost always converges to a pure Nash equilibrium within relatively few steps.  Second, our experiments further characterize producer specialization, showing that it occurs even under more complex settings than those in Section~\ref{sec:simple_eq}. We then study the effect of the temperature parameter in the softmax content-serving rule on both producer specialization and utility. Specifically, we demonstrate that the degree of producer specialization is monotonic with respect to temperature: higher temperatures (representing greater exploration in the content shown to users) incentivize more homogeneous content production, while lower temperatures (favoring the most relevant content for each user) encourages specialization, with producers focusing on distinct types of content. Lastly, we show that both producer and user utilities decrease monotonically with increasing temperature. A low softmax temperature yields the highest utility for both producers and users, emphasizing the benefits of selecting lower temperatures in the softmax content-serving rule.
\end{itemize}

\subsection{Related work} The study of strategic behavior and incentives among producers in recommender systems has seen much interest recently. These interactions can broadly be classified into two types of games: exposure-based, where producers are rewarded for maximizing the reach of their content, and engagement-based, where producer rewards depend not only on reach but also on how well the recommended content aligns with user preferences.
\paragraph{Exposure games} The seminal works of \citet{raifer2017information, basat2017game, ben2018game, ben2019convergence, ben2019recommendation, ben2020content} introduce game-theoretic models of competition among producers aiming to maximize exposure. While these early studies constrain content creators to finite strategy spaces by requiring them to select from a pre-specified, finite catalog, more recent works—such as \citet{meena-ss, hron2022modeling}, which are most closely related to ours—have relaxed this assumption by focusing on creators with continuous strategy sets. 

\paragraph{Engagement games} The works of \citep{topk-recsys, immorlica2024clickbait, huttenlocher2024matching} model producer rewards based on engagement. \citet{topk-recsys} study social welfare in scenarios where content creators compete without the mediation of a central recommender, and focus on bounding the Price of Anarchy\footnote{The ratio between the optimal welfare to that at the decentralized equilibrium\citep{koutsoupias1999worst}}. \citet{immorlica2024clickbait} examine engagement games from a different perspective, exploring the trade-off between producing high-quality content and gaming the recommender system by creating low-quality, "clickbait" content. \citet{huttenlocher2024matching} model the problem as a two-sided marketplace with departing users and creators and show that maximizing total engagement in such a setting is NP-hard.

\paragraph{Mechanism design} 
Works on Exposure and Engagement games focus on characterizing producer competition and equilibrium under a fixed content-serving rule and producer reward. In contrast, \cite{yao2023rethinking, yao2024user} adopt a mechanism design perspective, designing rewards and serving rules incentivizing equilibria that have high social welfare. \cite{hu2023incentivizing} explore a similar setting but model the platform as a linear contextual bandit.

\paragraph{User dynamics} Our work studies adaptive producers who change their content vectors to maximize user engagement. The user preferences are static, consistent with the literature on Exposure and Engagement games. Another line of work by \citep{DeanPrefdyna,PrefAmplificationMeta} studies shifting user preferences based on the content recommended to them. \citet{usercreatordual} adopt this user preference shift model but also model evolving producers, providing conditions for user polarization, though their producer evolution is not game-theoretic.

A table providing a more detailed comparison of our paper to related work along various axes (such as the nature of producer reward, type of equilibrium and dynamics) is available in Appendix \ref{app:relatedwork-table}. 

\section{Our model}\label{sec:model}
We consider an \emph{engagement game} between $n$ producers on an online platform. The producers must decide what type of content to produce in a $d$-dimensional space to maximize the engagement from users. This space is generally not what one may think of as the original feature space, where features are defined, e.g., as different genres that a producer or user may care about. Rather, these features are embedded features that are the results of a matrix factorization algorithm\footnote{MF learns latent representations of users and movies which are then used to predict user preferences and ratings. Our results however do not depend on the specifics of the algorithm used to obtain this embedding.} whose goal is to learn representative ``directions'' of the recommendation problem, as in~\cite{hron2022modeling,meena-ss}. The online platform then uses their \rs to recommend content to users as a function of how well producer content matches user preferences. More formally:

\paragraph{Producer model} We have $n$ producers on the platform. Each producer $i$ \emph{chooses} a content vector $s_i$ from the set $\cS \defn \{s:~s \in \Rorth^d,~\Vert s \Vert_1 \leq 1\}$. We note that we focus on the $\ell_1$-norm in order to model the relative amount of weight that each producer has on each embedded feature. We let $\Delta\cS$ be the set of probability distributions over $\cS$. Letting $s_i(f)$ denote the $f$-th entry of content vector $s_i$ for feature $f \in [d]$, one can interpret $s_i(f)$ as the fraction of producer $i$'s content that targets embedded feature $f$. 

\paragraph{User model} We have $K$ users on the platform. Each user $k \in [K]$ is described by a preference vector in $\C \defn \{c:~c \in \Rorth^d,~\Vert c \Vert_1 \leq 1\}$. We note that $c_k$ describes user preferences in the form of the amount of weight they attribute to each embedded feature. The more weight on the feature, the more utility the user derives from seeing content that aligns with said feature. We assume that the utility of a user with preferences $c_k$ who faces content $s$ is given by $c_k^\top s$.

\paragraph{Recommender system's content-serving rule} The platform uses a recommender system (RS) to decide which producer's content to show to which user. Denote $\vec{s} = (s_1,\ldots,s_n) \in \cS^n$ the full profile of producers' production choices. The RS shows producer $i$'s content $s_i$ to a user with preferences $c_k$ with probability $p_i(c_k,\vec{s}) = p_i(c_k, s_i,\sminus)$, where $\sminus$ denotes the rest of the producers. We call this probability the \emph{content-serving rule}. In this paper, we consider two content-serving rules: 
\begin{enumerate}
        \item The \emph{softmax} content-serving rule is, as in ~\cite{hron2022modeling, chen2019top}:
        \begin{align}
        \label{eq: sigmoid_p}
        p_i(c,s_i,\sminus) &\triangleq \frac{\expnt{\frac{c^\top s_i}{\tau}}}{\sum_{j=1}^n \expnt{\frac{c^\top s_j}{\tau}}}
        \end{align}
        where $\tau$ denotes the softmax temperature. A low temperature corresponds to greedier serving (i.e., only the best fitting producer's content is shown to the user), whereas a high temperature corresponds to adding more randomness to the serving (``worse'' producers may still have their content shown to the user, albeit with lower probability). The limit $\tau \to 0$ corresponds to a hard maximum, i.e., the producer whose content is best aligned, namely producer $\argmax_{j \in [n]} c^\top s_j$, is shown to user $c$. 
    \item The \emph{top-$k$ softmax} content-serving rule first selects the top-$k$ producers with the highest alignment scores $c^\top s_j$, where $j \in [n]$. Among the selected $k$ producers, we then apply the softmax function with temperature $\tau$:
        \begin{align}
        \label{eq: topk_sigmoid_p}
        p_i(c,s_i,\sminus) &\triangleq 
        \begin{cases} 
        \frac{\expnt{\frac{c^\top s_i}{\tau}}}{\sum_{j \in  \mathcal{K}} \expnt{\frac{c^\top s_j}{\tau}}} & \text{if } i \in \mathcal{K}, \\
        0 & \text{otherwise},
        \end{cases}
        \end{align}
    where $\mathcal{K} = \left\{j \in [n] : c^\top s_j \text{ is among the top $k$ values of } \{c^\top s_1, \ldots, c^\top s_n \} \right\}$.
    When $k = n$, the top-$k$ softmax rule reduces to the regular softmax rule defined in \eqref{eq: sigmoid_p}, and when $k = 1$, it's defined as the \emph{greedy serving rule}.
    \item The \emph{linear-proportional} content-serving rule (\cite{Luce1977215} choice axiom), where each producer's content $s_i$ is shown to user $c_k$ with a probability directly proportional to $c_k^\top s_i$.
        \begin{align}\label{eq: linear_p}
        p_i(c,s_i,\sminus) \triangleq 
        \begin{cases}
        \frac{c^\top s_i}{\sum_{j=1}^n c^\top s_j} &\text{if}~~c^\top s_i > 0,\\
        0  &\text{if}~~c^\top s_i = 0.
        \end{cases}
        \end{align}
    Note that the linear serving rule is well-defined even when $c^\top s_j$ is zero for all producers $j$. We use this rule as a baseline to compare performance to the typical softmax-based rule, and as an alternative rule to derive theoretical insights (our experiments in Section \ref{sec:experiments} show that theoretical insights for the linear-content serving rule in fact extend to the softmax rule).
    \item The \emph{round-robin serving rule} serves producers in a cyclic order: In the first round, all users are shown producer $1$'s content, in the second round producer $2$ and so on. Formally, the serving probability for user $c$ in serving round $r$ is defined as:
    \begin{align}
    \label{eq: round_robin_p}
    p_i^r(c, s_i, \sminus) &\triangleq 
    \begin{cases} 
    1 & \text{if } i = (r-1 \mod n) + 1, \\
    0 & \text{otherwise}.
    \end{cases}
    \end{align}
\end{enumerate}

\paragraph{Producer utilities} One of our contributions is to characterize the utility of a producer using \emph{engagement}, rather than just \emph{exposure}. In exposure games, the utility of a producer is simply the probability that this producer's content is shown to a user. In contrast, with \emph{engagement}, the utilities incorporate an additional term that measures \emph{how much a user engages with the content once this content is shown to them}. Formally, we assume that a producer $i$ who successfully shows content $s_i$ to user $k$ with preferences $c_k$ derives utility $s_i^\top c_k$ from that user. This captures the fact that a user whose preferences are better aligned with the producer's content are more likely to spend more time engaging with that content. Formally, we define the \emph{expected engagement utility} for producer $i$ as 
\begin{align}
    u_i({s}) = u_i(s_i,\sminus) 
    \triangleq \sum_{k=1}^K p_i(c_k, \vec{s}) \cdot c_k^\top s_i.\label{eq:prod-util}
\end{align}
Note that this \emph{expected} utility is reweighted by the probability of producer $i$ showing $s_i$ to user $k$, as $i$ derives no utility from user $k$ if said user does not see his content in the first place. The total producer utility $U_p$ is then defined as
\begin{align*}
U_p \triangleq \sum_{i=1}^n \sum_{k=1}^K p_i(c_k, \vec{s}) \cdot c_k^\top s_i.
\end{align*}

\paragraph{User utilities} We similarly define the utility for a user with embedding $c_k$ as its engagement in expectation across all producers i.e., $\sum_{i = 1}^n p_i(c_k,\vec{s}) \cdot c_k^\top s_i$. The total user utility $U_u$ is defined as
\begin{align*}
U_u \triangleq \sum_{k=1}^K \sum_{i=1}^n  p_i(c_k, \vec{s}) \cdot c_k^\top s_i.
\end{align*}

\begin{remark} The value of the total producer utility $U_p$ and the total user utility $U_u$ are equal, and the average producer utility and the average user utility are equal up to a multiplicative factor.
\end{remark}

\begin{remark}
\label{rem:round-robin-conv}
For the round-robin serving rule \eqref{eq: round_robin_p}, it is easy to see that the utility for a producer $i$ is zero if it is not being served. When it is served, the utility is given by $\max_{s_i \in S} \sum_{k=1}^K c_k^\top s_i = \norm{\sum_{k=1}^K c_k}_\infty$,
which is achieved by setting $s_i$ to the basis vector corresponding to the largest weight.
\end{remark}

\section{Equilibrium Structure}\label{sec:eq_structure}
In this section, we derive our main structural result for equilibrium in engagement games: namely, we show that at equilibrium, \emph{each producer prefers targeting a single embedded feature at a time.} 

Our first main assumption is that for all users, their features are strictly positive: 
\begin{assumption}\label{as:positive_c}
For every user $k$ and feature $f$, we have $c_k(f) > 0$. 
\end{assumption}
% \kri{also for the adversarial reviewer: non-negative embeddings are common in topic modelling LDA \citet{blei2003latent}, TF-IDF, BERTopic\cite{grootendorst2022bertopic}}. 
This assumption holds in practice when user representations are obtained via \emph{Non-Negative Matrix Factorization} (NMF)~\citep{lee2000algorithms, luo2014efficient} as is observed in~\citet{meena-ss, hron2022modeling} and in our experiments in Section \ref{sec:experiments}. We additionally make an assumption on the data distribution, guaranteeing that user preferences $c_1, \ldots, c_k$ are non-trivial:

\begin{assumption}\label{as:nontrivial_c}
For all pairs of production strategies $s,~s' \in \cS$ such that $s \neq s'$, there must exist at least one user $k$ such that $c_k^\top s \neq c_k^\top s'$.
\end{assumption}

% \textcolor{red}{What we actually need as an assumption below}

% \begin{assumption}\label{as:nontrivial_c}
% Pick any producer $i$. For all pairs of production strategies $s_i,~s_i' \in \cS$ such that $s_i \neq s'_i$, and for all (non-zero) strategy $s_{-i}$ for the other players, there must exist at least one user $k$ such that $c_k^\top s \neq c_k^\top s'$ and such that $\sum_{j \neq i} c_k^\top s_{j} > 0$. 
% \end{assumption}

This is a mild assumption that states that there is enough diversity in user preferences. Equivalently, this assumption states that user preferences are non-trivial and diverse enough such that $span(c_1,\ldots,c_K) = \mathbb{R}^d$, i.e., the user preferences span \footnote{Indeed, there then exists a subset of size $d$ of $(c_1,\ldots,c_K)$ that forms a basis for $\mathbb{R}^d$. If $s \neq s'$, they must differ in at least one coordinate in this basis, so there must exist $c_k$ such that $c_k^\top s \neq c_k^\top s'$} all of $\mathbb{R}^d$. If the user preferences do not span all of $\mathbb{R}^d$, this means that some latent features are redundant. In practice, one may work with a
reduced embedding dimension and perform a new matrix factorization until Assumption~\ref{as:nontrivial_c} holds.

We are now interested in understanding properties of the Nash equilibria (NE) of our engagement-based content production game, defined as the game where each producer decides which content to produce to maximize their utility. %\jw{I think the defintion of the game is implicit? Say explicitly that each producer's action is to choose what content to produce}
Nash equilibria are a classical concept for solving games~\citep{nash1951}. We apply the standard definition to our formulation as follows.

\begin{definition}[Nash Equilibrium]
For any producer $i$, the strategy $s_i^*\in \cS$ that solves $u_i(s_i^*,\sminus) = \max_{s_i\in \cS}~u_i(s_i,\sminus)$ is called a \emph{best response} to $\sminus$. A strategy profile $\left(s_1^*,\ldots,s_n^*\right) \in \cS^n$ is a pure-strategy Nash equilibrium (pure NE) if and only if for every producer $i \in [n]$,
\[
u_i(s_i^*,\sminus^*) = \max_{s_i\in \cS}~u_i(s_i,\sminus^*).
\]
A strategy profile $\left(D_1^*,\ldots,D_n^*\right) \in \Delta\cS^n$, where $\Delta\cS^n$ denotes the probability simplex over $\cS^n$, is a mixed-strategy Nash equilibrium (mixed NE) if and only if for every producer $i \in [n]$,
\begin{align*}
\E_{s_i \sim D_i^*,s_{-i} \sim D_{-i}^*}\left[u_i(s_i,\sminus)\right]
 = \max_{D_i\in \Delta \cS}~\E_{s_i \sim D_i,s_{-i} \sim D_{-i}^*} \left[u_i(s_i,\sminus)\right].
\end{align*}
\end{definition}

Under either pure or mixed Nash equilibrium, all producers best respond to each other and do not want to change their strategy: i.e., each producer maximizes its utility and gets the best utility it can by playing the Nash, assuming all remaining producers also play the Nash.% The only difference between a pure-strategy and a mixed-strategy NE is that the latter allows any given producer to randomize their strategy over $\cS$. 

We now characterize the equilibria of our game, under both content serving rules of Equation~\eqref{eq: sigmoid_p} and Equation~\eqref{eq: linear_p}, showing that producers prefer to focus on a single embedding dimension at a time:

\begin{theorem}\label{thm:pure_eq} 
Suppose Assumptions~\ref{as:positive_c} and~\ref{as:nontrivial_c} hold. Let $\mathcal{B} \defn (e_1,\ldots,e_d)$ be the standard basis for $\mathbb{R}^d$, where each $e_j$ is the unit vector with value $1$ in coordinate $j \in [d]$ and $0$ in all other coordinates. Under both types of content-serving rules, if there exists a NE, any pure strategy for producer $i$ must satisfy $s_i^*\in \mathcal{B}$, and any mixed strategy must be a distribution supported on $\mathcal{B}$ in this equilibrium.
\end{theorem}
% \begin{proof}[Proof Sketch for Theorem~\ref{thm:pure_eq}]
% First, we show that at equilibrium, the utility function of each producer is strictly convex in their strategy. %The strictness is derived from the fact that at equilibrium, no producer plays a trivial strategy in which $s_i = 0$. |
% Then, given strict convexity, each producer's utility must be maximized on the vertices of the simplex $(e_1,\ldots,e_d)$. The full proof is in Appendix~\ref{app:proof_theorem}.
% \end{proof}
% The full proof is in Section~\ref{app:proof_theorem}. 
Informally, the theorem shows that at equilibrium, producer strategies are supported on the standard unit basis rather than on the entire simplex $\cS = \left\{s:~\Vert s \Vert_1 \leq 1,~s \in \mathbb{R}^d_{\geq 0}\right\}$. Each producer focuses on a single feature in the embedded space. Conceptually, the producers play a feature selection game where they trade-off i) choosing the ``best'' features (those that align best with users' preferences) to improve the reward when showing content to users with ii) choosing potentially sub-optimal features to avoid competition with other producers over top features and improve the chance of showing content to users.

The remainder of the paper studies how the choice of content-serving rule affects the trade-off between effects i) and ii). i) pushes for more homogeneous content production while ii) promotes specialization. Section~\ref{sec:simple_eq} provides some theoretical insights on this trade-off, and Section~\ref{sec:experiments} provides experimental results for both the linear-proportional and the softmax serving rule.

\subsection{Proof of Theorem~\ref{thm:pure_eq}}\label{app:proof_theorem}
\paragraph{Preliminary properties of producers' utilities} We start by noting the convexity properties of producers' utilities. Everywhere in this proof, we denote $S(c,\sminus) =  \sum_{j \neq i}^n c_k^\top s_j$.

\begin{claim}[Convexity for \emph{linear-proportional} serving rule]
The function
\begin{align*}
f_{k,\sminus}(x) = \frac{x^2}{x + \sum_{j \neq i}^n c_k^\top s_j}
\end{align*}
is convex in $x$ on domain the $\Rpos$. Further, suppose $\sum_{j \neq i}^n c_k^\top s_j > 0$. Then it is strictly convex in $x$ on the domain $\Rpos$.
\end{claim}

\begin{proof}
The first-order derivative is given by $f_{k,\sminus}'(x) = \frac{x (2S(c,\sminus) + x)}{(S(c,\sminus)+x)^2} \geq 0$. The second order derivative of $f$ is given by $f_{k,\sminus}''(x) = \frac{2S(c,\sminus)^2}{(x + S(c,\sminus))^3}$. Note that on $x \in \Rpos$, $f_{k,\sminus}''(x) \geq 0$, with $f_{k,\sminus}''(x) > 0$ when $S(c,\sminus) > 0$. This concludes the proof. 
\end{proof}

\begin{claim}[Convexity for \emph{softmax} serving rule] The function
\[
f_{k,\sminus}(x) = \frac{x \expnt{x/\tau}}{\expnt{x/\tau} + \sum_{j \neq i}^n\expnt{c_k^\top s_j/\tau}}
\]
is strictly convex in $x$ for $|x| \leq \tau \log(S_{exp})$.
\end{claim}
In practice, we note that we expect this condition on $x$ to hold when the game is large enough. Indeed, it is equivalent to $S_{exp} > \expnt{x/\tau}$. When the number of producers grows, $S_{exp}$ also grows while $\expnt{x/\tau}$ remains bounded by $\expnt{1/\tau}$ (we restrict attention in this entire proof to $x$ representing an inner product of the form $c^\top s$, which we know is in $[0,1]$).

\begin{proof}
Let us overload the $S$ notation and write $S_{\exp} = \sum_{j \neq i} \expnt{c_k^\top s_j/\tau}$. The first-order derivative is given by 
\begin{align*}
f_{k,\sminus}'(x) = \frac{\left(1 + x/\tau\right) \expnt{x/\tau} S_{\exp}  + \expnt{2x/\tau}}{(\expnt{x/\tau} + S_{\exp})^2}
\end{align*}
The second order derivative is then given by
\begin{align*}
 f_{k,\sminus}''(x)  
 &= \frac{S_{exp} \expnt{x/\tau}}{\tau^2 (S_{exp} + \expnt{x/\tau})^3} \cdot  (2 S_{exp} \tau + S_{exp} x + 2 \tau \expnt{x/\tau} - x \expnt{x/\tau}))
 \\& = \frac{S_{exp} \expnt{x/\tau}}{\tau^2 (S_{exp} + \expnt{x/\tau})^2} \cdot \left(2 \tau + x \frac{S_{exp} - \expnt{x/\tau}}{S_{exp} + \expnt{x/\tau})} \right)
\end{align*}

Notice that this is strictly positive so long as $S_{exp} > \expnt{x/\tau}$, concluding the proof. 
\end{proof}
Now, we note that for every producer $j$, in each strategy $s_j$ supported in a mixed strategy profile at equilibrium, we must have $s_j > 0$. Indeed, if any producer sets $s_j = 0$ and produce no content, they get a utility of $0$; however, setting $s_j > 0$ leads to $s_j^\top c_k > 0$ for all users $k$, by Assumption~\ref{as:positive_c}, and yields strictly positive utility. For a mixed strategy profile to constitute an equilibrium, every action on the support of each player's strategy must have the same utility, and this utility then has to be strictly positive, so it must be that $s_j > 0$ on the entire mixed strategy's support. Then, $S(c,\sminus) > 0$ on any strategy profile on the support of a mixed Nash equilibrium. Therefore, we can restrict attention to $\mathcal{S}' = \mathcal{S}/\{0\}$. We know that $f_{k,\sminus}(x)$ is then strictly convex in $x$.
Now, let's examine the function $g_{k,\sminus}(s_i) = f_{k,\sminus}(c_k^\top s_i)$. Clearly, $g_{k,\sminus}$ is convex and for all $k$, 
\[
g_{k,\sminus}(\lambda s_i + (1-\lambda) s_i')  \leq \lambda g_{k,\sminus}(s_i) + (1-\lambda) g_{k,\sminus}(s_i').
\]
Further, pick any $\lambda \in (0,1)$ and $s_i \neq s_i'$. There must exist $k$, by Assumption~\ref{as:nontrivial_c}, such that $c_k^\top s_i \neq c_k^\top s_i'$. Then, we have that, for that $k$, 
\begin{align*}
g_{k,\sminus}(\lambda s_i + (1-\lambda) s_i') 
& = f_{k,\sminus}(\lambda c_k^\top s_i + (1-\lambda) c_k^\top s_i') 
\\& < \lambda f_{k,\sminus}(c_k^\top s_i) + (1-\lambda) f_{k,\sminus}(c_k^\top s_i') 
\\&= \lambda g_{k,\sminus}(s_i) + (1-\lambda) g_{k,\sminus}(s_i'),
\end{align*}
where the inequality follows by strict convexity of $f_{k,\sminus}(x)$. Summing over all $k$'s, we then get that
\[
\sum_{k=1}^K g_{k,\sminus}(\lambda s_i + (1-\lambda) s_i') < \lambda \sum_{k=1}^K g_{k,\sminus}(s_i) + (1-\lambda) \sum_{k=1}^K g_{k,\sminus}(s_i')
\]
as at least one of the inequality in the sum has to be strict. This shows that while each $g_{k,\sminus}$ is not necessarily strictly convex, $\sum_k g_{k,\sminus}$ is. Now, note that both in the linear case as per Equation~\ref{eq: linear_p} and in the softmax case as per Equation~\ref{eq: sigmoid_p}, we have that producer $i$'s utility for playing $s_i$, under a mixed strategy profile $\sminus \sim \cD$ for the remaining agents, is given by 
\[
u_i(s_i,\cD) 
= \E_{\sminus \sim  \cD} \left[\sum_{k=1}^K  g_{k,\sminus}(s_i)\right]
\]

Then, $u_i(s_i,\cD)$ is  also strictly convex in $s_i$ (this follows immediately by writing the definition of strict convexity and by linearity of the expectation). Finally, now, suppose by contradiction that producer $i$'s best response is not a unit vector $e_f$. Then there exists $f$ with $0 < s_i(f) < 1$. In this case, note that
\[
u_i(s_i,\cD) 
= u_i \left(\sum_f s_i(f) e_f, \cD \right) 
< \sum_f s_i(f) u_i(e_f,\cD),
\]
by strict convexity of $s_i \rightarrow u_i(s_i,\cD)$. This is only possible if there exists $f'$ such that $u_i(s_i,\cD)  < u(e_{f'},\cD)$.
This concludes the proof. 

Below, we provide our heuristic for computing Nash equilibria, under Algorithm~\ref{alg:bestrep_dynamics}. Our heuristic relies on the structural result of Theorem~\ref{thm:pure_eq}: if a pure equilibrium exists, then it must be supported on the standard basis $\mathcal{B}$, which simplifies the best-response computation.

\begin{algorithm}[!h]
\caption{Best Response Dynamics for Pure Equilibrium Computation~\label{alg:bestrep_dynamics}}
\textbf{Inputs}: User embeddings $(c_1,\ldots,c_K)$. Utility $u_i(s_i,\sminus)$ for producer $i$. Max iterations $N_{max}$.\\
\textbf{Output}: Pure Nash equilibrium of the engagement game.
\begin{algorithmic}
\STATE Initialize termination variable $fin = 0$ and iteration variable $iter = 0$. Initialize producer $i$'s strategy $s_i$ uniformly at random in $\mathcal{B} = (e_1,\ldots,e_d)$.
\WHILE {$fin = 0$ and $iter < N_{max}$} 
    \STATE Produce a random permutation vector of producers $1$ to $n$. 
    \FOR {each producer $i$ in the above permutation}
        \STATE Set $fin = 1$;
        \STATE Compute $s_i^* = \arg\max_{s_i \in \mathcal{B}} u_i(s_i,s_{-i})$;
        \IF {$u_i(s_i^*,s_{-i}) > u_i(s_i,s_{-i})$}
            \STATE $s_i = s_i^*$;
            \STATE Set $fin = 0$ and exit the for loop.
        \ENDIF
    \ENDFOR
\ENDWHILE
\STATE \textbf{return} $(s_1,\ldots,s_n)$ if $fin = 1$, and $\bot$ if $fin = 0$.
\end{algorithmic}
\end{algorithm}

\section{Equilibria under Simplified User Behavior}\label{sec:simple_eq}
To provide theoretical insights about engagement equilibria,
we first consider a simple, special case of our framework where each user is single-minded, and is only interested in a single type of content. We also focus on the linear-proportional content serving rule of Equation~\eqref{eq: linear_p} for tractability.

\begin{assumption}[Single-minded users]\label{as:single-minded}
For any user $k$, $c_k = e_f$ for some $f \in [d]$, where $(e_1,\ldots,e_d)$ is the standard basis.
\end{assumption}
 
Note that \emph{we only make this assumption in the current section}, and that this is an assumption on \emph{user and not producer behavior}. This simplified assumption and setting allow us to derive our first insights towards characterizing the equilibria of our recommender system, and how producers decide what type of content to produce as a function of the total user weight on each feature. Our experiments in Section~\ref{sec:experiments} show that the insights we derive in this simple single-minded user setting hold for general types of users under the right choice of content-serving rule.

Under this assumption, we provide a simplified characterization of the equilibria of our game. We consider equilibria supported on the standard basis $\mathcal{B}$, following the insights of Theorem~\ref{thm:pure_eq}, and let $(m_1,\ldots,m_d)$ be the number of users interested in content type $e_1,\ldots,e_d$, respectively. We assume $m_f > 0$ without loss of generality; otherwise, feature $f$ brings utility to no user nor producer, and can be removed. On the procedure side, we use the notation $(\numprod_1,\ldots,\numprod_d)$ to denote an aggregate strategy profile where for all $f \in [d]$, a number $\numprod_f$ of producers pick action $e_f$.

\begin{lemma}\label{lem:proportional}
Suppose Assumption~\ref{as:single-minded} holds. For all $f \in [d]$, let $m_f > 0$ be the number of users with $c = e_f$. Under the linear-proportional content serving rule of Equation~\eqref{eq: linear_p}, there exists a pure NE supported on $\mathcal{B}$ and given by $(\numprod_1,\ldots,\numprod_d)$ if and only if
\begin{align}\label{eq:deviation}
\frac{\numprod_f}{m_f} \leq \frac{\numprod_{f'} + 1}{m_{f'}}\qquad\text{for all } f,f' \in [d].
\end{align}
\end{lemma}

The full proof is given in Appendix~\ref{app:proof_simple}. 

\paragraph{Interpretation of the equilibrium conditions}
Let $\delta_f \triangleq \frac{m_f}{\sum_{f'=1}^d m_{f'}}$ be the fraction of users that are interested in feature $f$. Consider strategy profile $\numprod_f = \delta_f n$, and assume $n_f$ is integer\footnote{When $\delta_f n$ is not integral, we can instead round carefully so that $\numprod_f$'s still add to $n$, and the proposed strategy remains an approximate NE, in that the equilibrium conditions are satisfied up to a small additive slack factor}. We show that this pure strategy profile is an equilibrium of our engagement game, by verifying that Condition\eqref{eq:deviation} in Lemma~\ref{lem:proportional} holds for this construction. To see this, note that the equilibrium condition is equivalent to
$
\frac{\numprod_f}{\delta_f} \leq \frac{\numprod_{f'} + 1}{\delta_f'}.
$
On the left-hand side, we have $\frac{\numprod_f}{\delta_f} =  n$. On the right-hand side, we have $\frac{\numprod_{f'} + 1}{\delta_{f'}} = \frac{\delta_{f'} n + 1}{\delta_{f'}} \geq n$, and the inequality always holds.

In this equilibrium, the number of producers that pick feature $f$ is (ignoring rounding) proportional to the number of users interested in feature $f$. This aligns with intuition:
at equilibrium, the probability with which each producer is recommended to a user should be the same across all features $f$; otherwise, a producer that deviates to a feature $f'$ with higher probability will be shown more often and gain more utility.
We note that~\citet{hron2022modeling} made a related observation in a different setting\footnote{They use exposure instead of engagement, different assumptions on user preferences, and the softmax content-serving rule}: they note that producer content aligns with the average user weight on each embedded feature, with the average producer content being $\bar{c} = \frac{1}{K} \sum_{k \in K} c_k$ at an approximate equilibrium.
However, we note that the insight of~\citet{hron2022modeling} remains different, in that all producers produce the same homogeneous content aligned with $\bar{c}$; we show specialization and heterogeneous content production, which we believe are commonplace in practice.

\section{Experiments} \label{sec:experiments}
We now provide experiments that expand our understanding of producer equilibria and utilities at equilibrium for general user incentives.

\subsection{Experimental setup} 

\paragraph{Synthetic Data} We generate three types of user distributions for our synthetic data. In the \emph{uniform distribution}, we generate user embeddings $c$ uniformly at random over the probability simplex. Hence, each feature is equally represented in the data in expectation. In the \emph{skewed distribution}, we first randomly sample positive weights $w_1 \leq \ldots \leq w_d$ (sampled from the probability simplex and then sorted). We then generate a uniformly distributed user embedding, but re-weight each feature $f$ by $w_f$ to obtain $c$. This creates a non-symmetric, skewed distribution where the total user weight for each feature is proportional to $w_f$, leading to differences across features. We also generate a \emph{sparse distribution} for which we generate user embeddings $c \in \mathbb{R}^d$ uniformly at random over the probability simplex and then apply an element-wise masking operation. This operation uses random boolean vectors $\in \{0,1\}^d$, with $90\%$ of its values being zero.
As per our modeling assumptions, we normalize user features $c$ to have $\ell_1$-norm $1$. All synthetic experiments use $K = 10,000$ users. We vary the dimension $d$ and the number of producers $n$. 

\paragraph{Real Data} Following prior work on producer-side competition, we use the NMF implementation in the \texttt{scikit-surprise} package~\citep{scikit-surprise} to obtain user embeddings for the MovieLens-100k dataset~\citep{movielens}; this dataset contains $100k$ ratings, $943$ users and $1682$ movies. We also run experiments with two other, larger-scale ratings datasets: AmazonMusic rating~\citep{amazon-ratings} and RentTheRunway clothing rating~\citep{rentrunway}, These datasets have around $840k$ and $100k$ unique users respectively. The insights from our experiments on these datasets turn out to be similar to the Movielens-100k dataset and are deferred to Appendix \ref{app:experiments}.

\paragraph{Equilibrium computation: a best-response based heuristic} We provide a simple heuristic based on best-response dynamics to compute pure-strategy Nash equilibria in Algorithm \ref{alg:bestrep_dynamics}. In each iteration, it goes through the list of producers in randomly sorted order. For each producer, it computes the best basis vector response of the producer. If the producer is already playing a best response, the algorithm goes to the next producer; otherwise, the algorithm notes the producer is not currently best responding, updates the producer's strategy to his best response, and starts the next iteration. If at the end of the for loop, we realize that all producers are playing a best response, we have found a NE, hence we stop and output the current strategy profile. If after $N_{max}$ iterations, we still have not found a NE, we output $\bot$ to signal that our dynamics did not converge. 

We will see that our heuristic terminates the large majority of the time
in our experiments; in that case, it must output a pure Nash Equilibrium, as the algorithm can only terminate when all producers best respond to each other.\footnote{We do not provide theoretical guarantees on the existence of pure NE or convergence of dynamics: in fact, Algorithm \ref{alg:bestrep_dynamics} does not converge in a few instances (Table \ref{table:expconv_table}), confirming that our game is not what is called a ``potential'' game. When a game is not potential, existence of pure NE and convergence of dynamics are not guaranteed, and equilibrium existence certification and computation is generally a computationally hard problem. Heuristics are often the best that we can hope for.}

\subsection{Experimental results}
\label{sec:conv_rate}

\paragraph{Convergence of Algorithm~\ref{alg:bestrep_dynamics} and rate of convergence}
For our engagement game, we consider $12$ different numbers of producers $(2, 5, 10, 20, \ldots, 100)$ and $6$ different embedding dimensions $(5, 10, 15, 20, 50, 100)$. Each of these $12 \times 6$ configurations is instantiated with $5$ random seeds to account for the randomness of the NMF algorithm.
Across all $360$ instances of the (producers, dimension, seed) configuration, the softmax content-serving rule with temperatures $\tau \in \{100, 10, 1, 0.1\}$ always converges to a unique Nash Equilibrium (NE) on the Movielens-100k dataset. Even with a low softmax temperature of $\tau=0.01$, Algorithm \ref{alg:bestrep_dynamics} still converges to a unique NE in a large number of instances. Additionally, for the top-$k$ softmax serving rule reducing the top-$k$ value—making the serving rule "greedier"—results in fewer converged instances. The linear-proportional serving rule converges in many instances, and the round-robin serving rule as highlighted in Remark~\ref{rem:round-robin-conv} always converges to the dimension with maximum weight.

\begin{table}[!h]
\centering
\begin{subtable}[t]{0.65\textwidth}
    \centering
    \begin{tabular}{lcccccc}
    \hline
    \textbf{Serving Rule} & $\tau = 100$ & $\tau = 10$ & $\tau = 1$ & $\tau = 0.1$ & $\tau = 0.01$ \\\hline
    (Full) Softmax        & 360          & 360         & 360        & 360         & 342          \\
    (Top-20) Softmax      & 253          & 247         & 232        & 182         & 326          \\
    (Top-10) Softmax      & 170          & 171         & 165        & 147         & 299          \\\hline
    \end{tabular}
    \caption{Converged instances for Softmax-based serving rules \label{tab:ne_convergence_temp}}
\end{subtable}%
\hfill
\begin{subtable}[t]{0.3\textwidth}
    \centering
    \begin{tabular}{lc}
    \hline
    \textbf{Serving Rule} & \textbf{Converged} \\\hline
    Greedy                & 269                \\
    Linear                & 357                \\
    Round-Robin           & 360                \\\hline
    \end{tabular}
    \caption{Converged instances for temperature-independent rules.\label{tab:ne_convergence_independent}}
\end{subtable}
\caption{Convergence to Nash Equilibrium (NE) across serving rules on the Movielens-100k dataset. (a) Softmax serving rules with varying temperature $\tau$ and top $k$ values, (b) Temperature-independent rules.\label{tab:nash-convergence-ml100k}}
\end{table}

Further, we examine the rate of convergence of the best response dynamics in Algorithm~\ref{alg:bestrep_dynamics} to a Nash Equilibrium. In Figure \ref{fig:num_iters_convergence} we plot the number of iterations  (averaged over $40$ runs) of Algorithm~\ref{alg:bestrep_dynamics}  with increasing number of producers and across varying embedding dimensions on the Movielens-100k dataset. Figures \ref{fig:convlinear-movielens100k} and \ref{fig:convsoftmax-movielens100k} plot the number of iterations to convergence with the linear and the softmax content-serving rule.  In Appendix \ref{app:numiters-mbody}, we provide figures for the synthetic datasets and observe insights similar to Figure \ref{fig:num_iters_convergence}.
\begin{figure}[!h]
\centering
\captionsetup[subfigure]{justification=centering}
\subfloat[the linear-proportional serving rule Movielens-100k]{
    \includegraphics[width=0.23\textwidth]{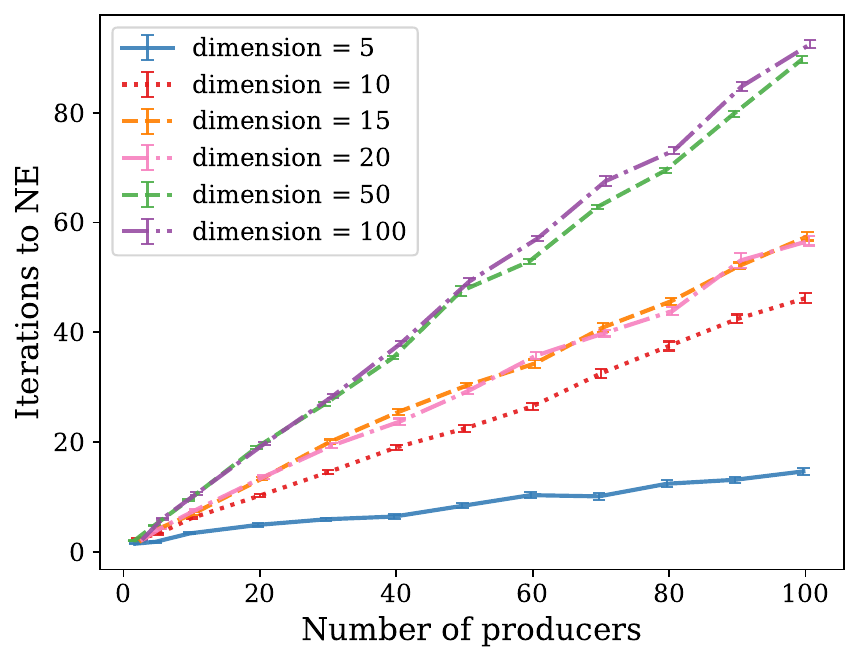}
    \label{fig:convlinear-movielens100k}
}
\hspace{0.05\textwidth}
\subfloat[the softmax serving rule Movielens-100k]{
    \includegraphics[width=0.23\textwidth]{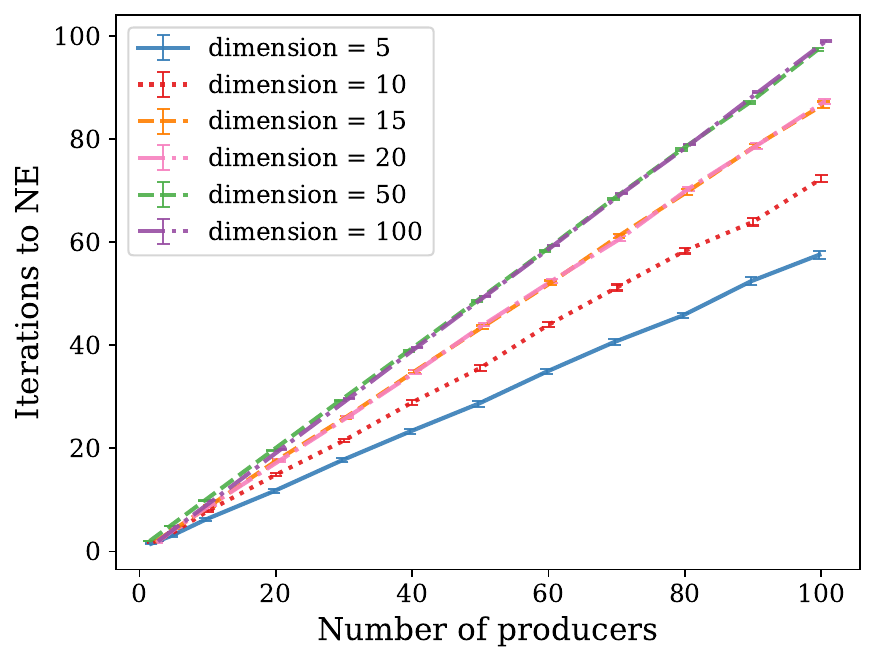}
    \label{fig:convsoftmax-movielens100k}
} 
\caption{Number of iterations of Algorithm~\ref{alg:bestrep_dynamics} until convergence to a Nash Equilibrium on the Movielens-100k dataset. The different curves represent different embedding dimensions in the game $d \in \{5,10,15,20,  {50, 100}\}$; the error bars represent standard error over $40$ runs.\label{fig:num_iters_convergence}}
\end{figure}

Based on the observations above, we conclude that Algorithm~\ref{alg:bestrep_dynamics} empirically seems to be a reliable and computationally efficient heuristic to find pure Nash equilibria for our engagement game, with performance scaling well with the number of producers. In contrast, a naive brute-force approach enumerating all best responses  takes exponential time over the number of producers.

In the following, we study how changing the softmax temperature affects producer specialization and the producer utility. We consider $5$ different values for the softmax temperature $\tau \in \{ 0.01, 0.1, 1, 10, 100\}$ and use the linear-proportional serving rule as a benchmark.

\paragraph{Equilibrium results} 
In Figures~\ref{fig:movielens-100k-udpd} and \ref{fig:skewed-udpd}, we 
highlight the impact of the content-serving rule and softmax temperature on the degree of specialization at the instance-level. We show a single (but representative) instance of the problem in each figure to provide a visual representation of producer specialization at equilibrium; our insights are consistent across our generated instances. 

\begin{figure}[!h]
\centering
\captionsetup[subfigure]{justification=centering}
\subfloat[Softmax $\tau = 100$]{
    \includegraphics[width=0.22\textwidth]{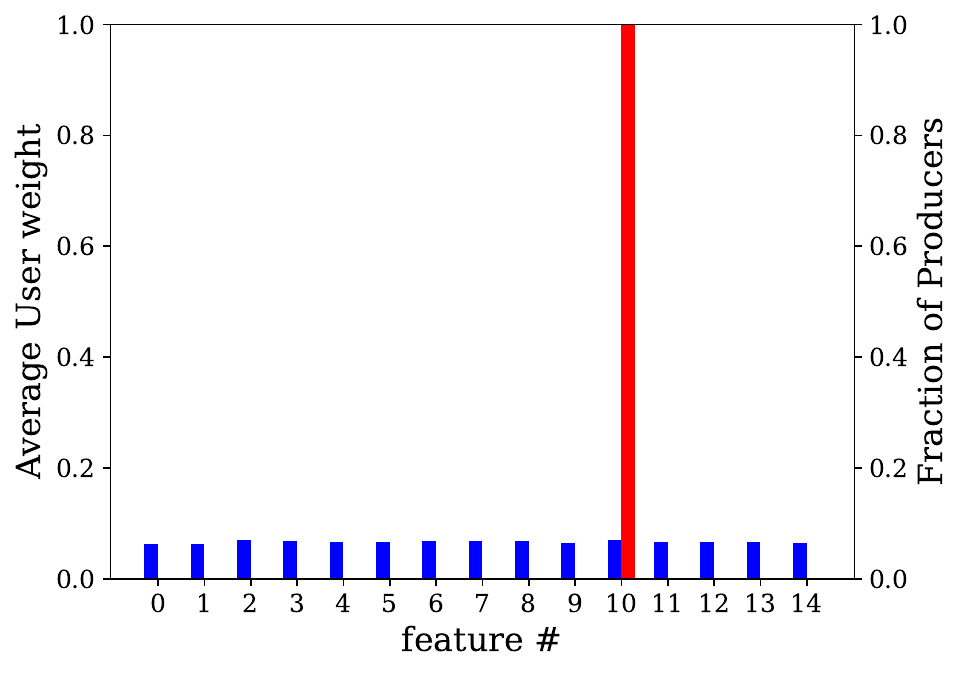}
\label{fig:ml100k-udpd-temp100}
}
\hfill
\subfloat[Softmax $\tau = 10$]{
    \includegraphics[width=0.22\textwidth]{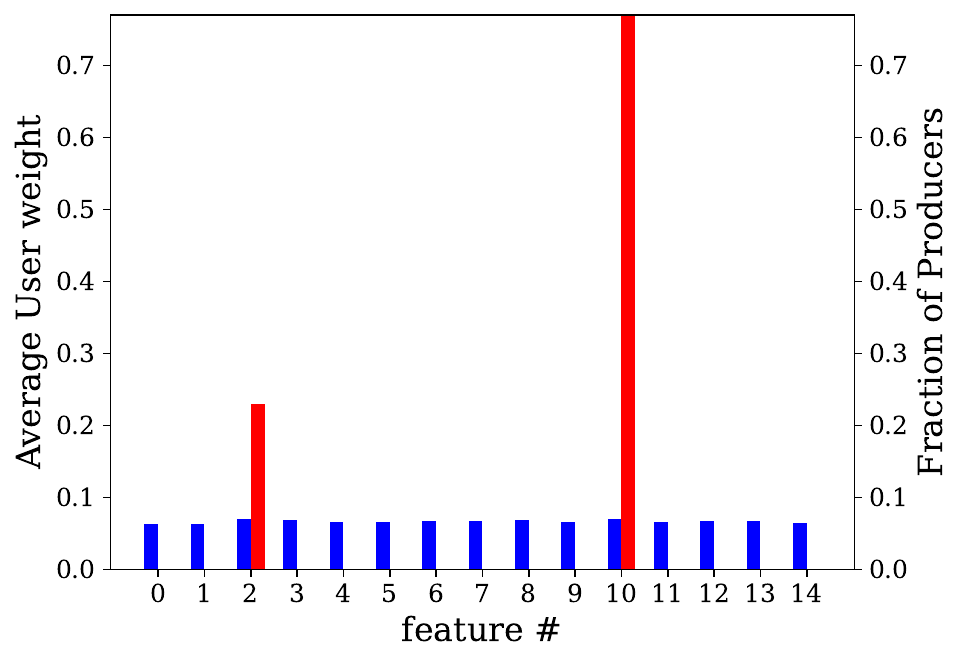}
\label{fig:ml100k-udpd-temp10}
}
\hfill
\subfloat[Softmax $\tau = 0.01$]{
    \includegraphics[width=0.22\textwidth]{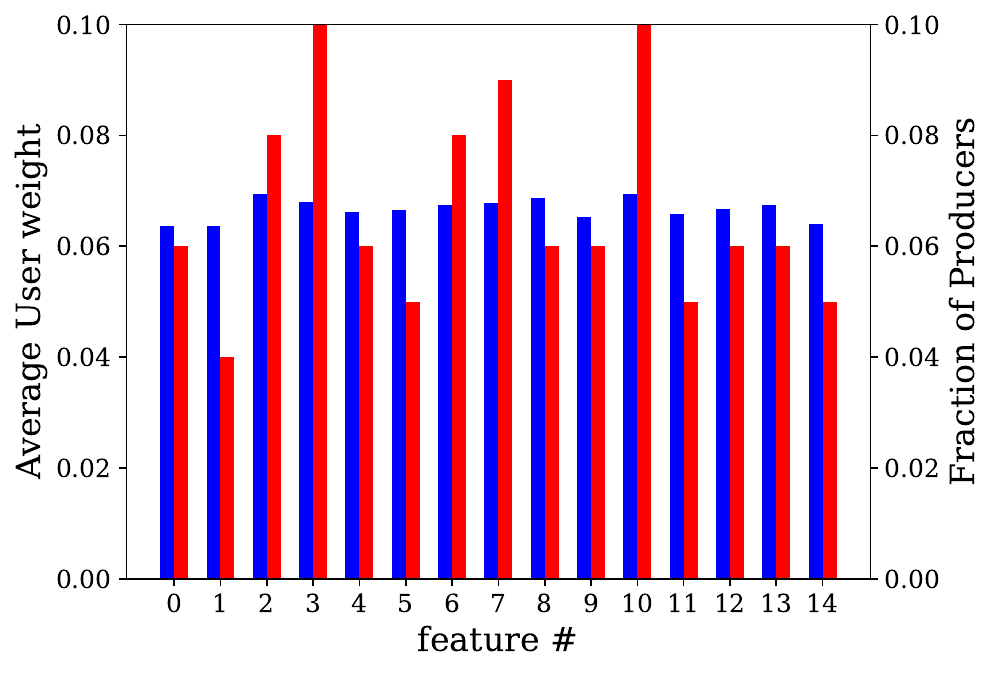}
\label{fig:ml100k-udpd-temp001}
}
\hfill
\subfloat[the linear-proportional serving rule]{
    \includegraphics[width=0.22\textwidth]{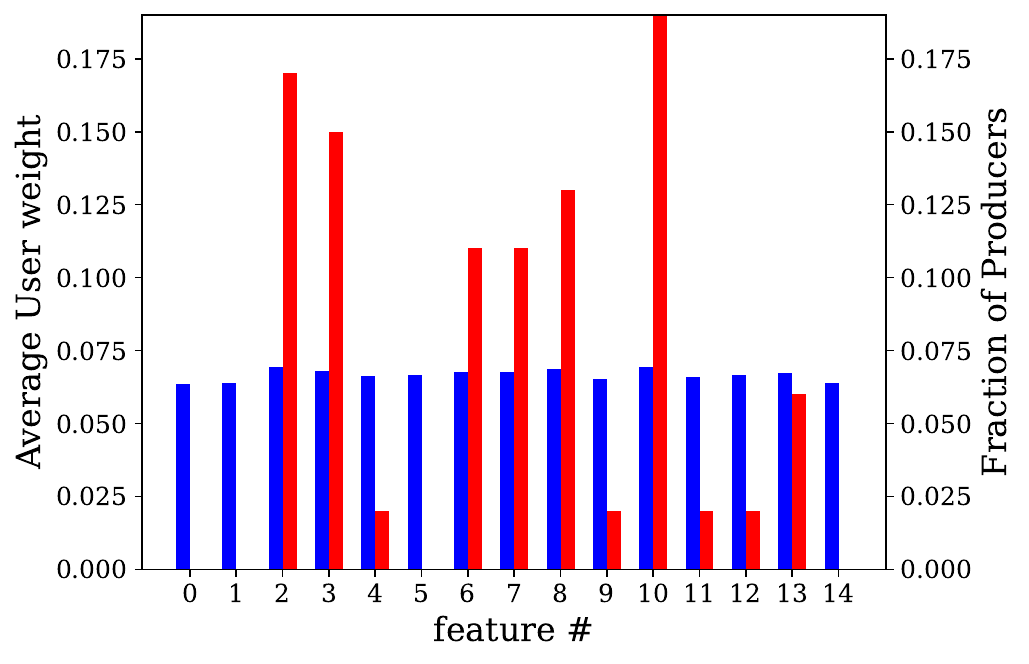}
\label{fig:ml100k-udpd-linear}
}
\caption{Average user weight on each feature (blue, left bar) and fraction of producers going for each feature (red, right bar) $n = 100$ producers, embedding dimension $d = 15$. Lower softmax temperature leads to more producer specialization. User embeddings obtained from NMF on MovieLens-100k.\label{fig:movielens-100k-udpd}}
\end{figure}

\begin{figure}[!h]
\centering
\captionsetup[subfigure]{justification=centering} 
\subfloat[Softmax $\tau = 1$]{
    \includegraphics[width=0.22\textwidth]{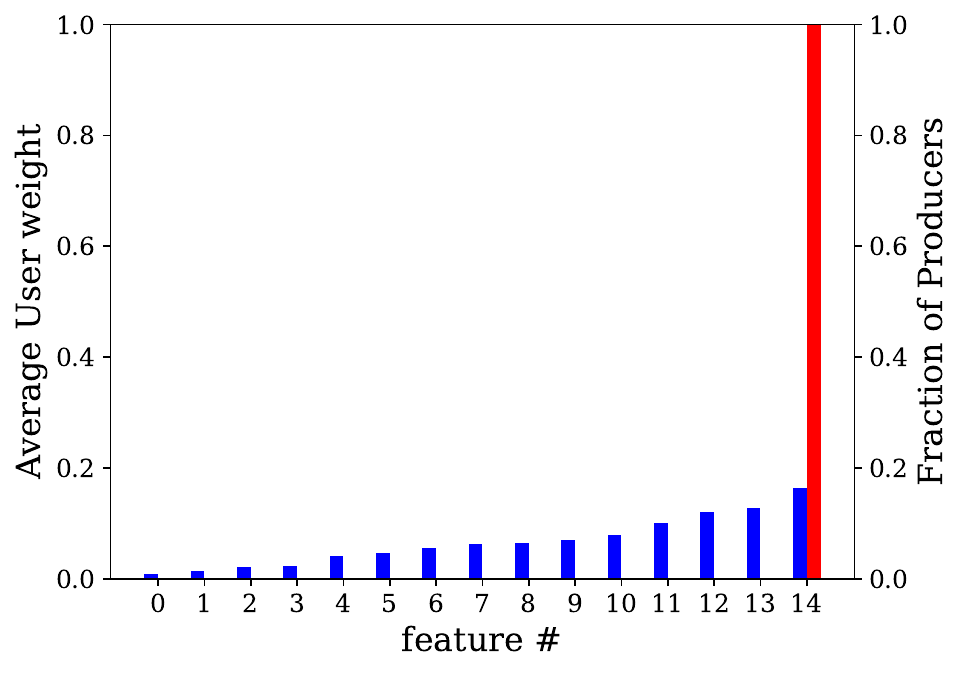}}
\hfill
\subfloat[Softmax $\tau = 0.1$]{
    \includegraphics[width=0.22\textwidth]{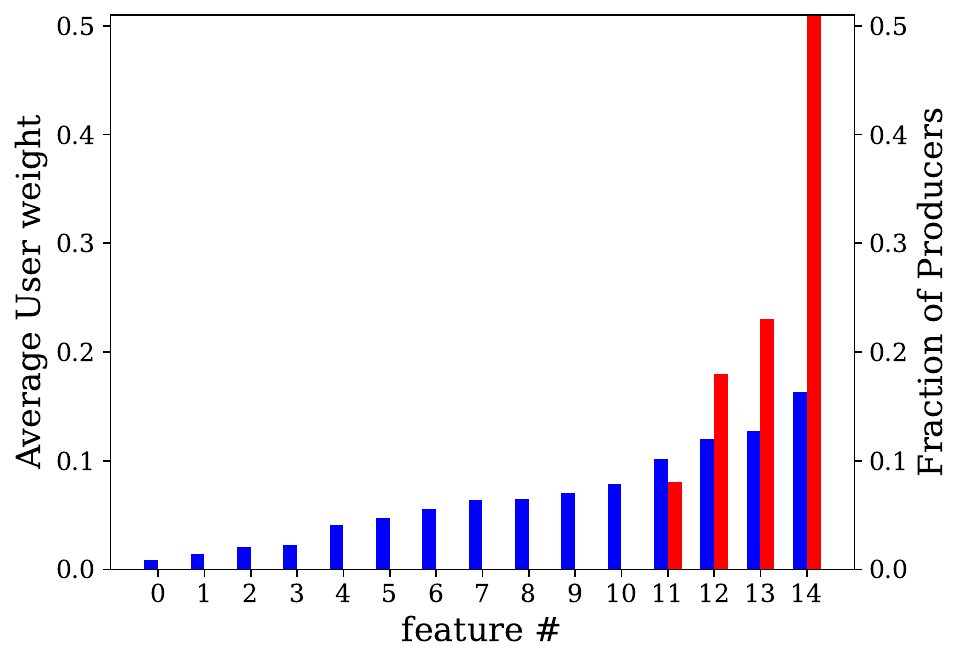}
    \label{fig:skewed-udpd-tau01}
}
\hfill
\subfloat[Softmax $\tau = 0.01$]{
    \includegraphics[width=0.22\textwidth]{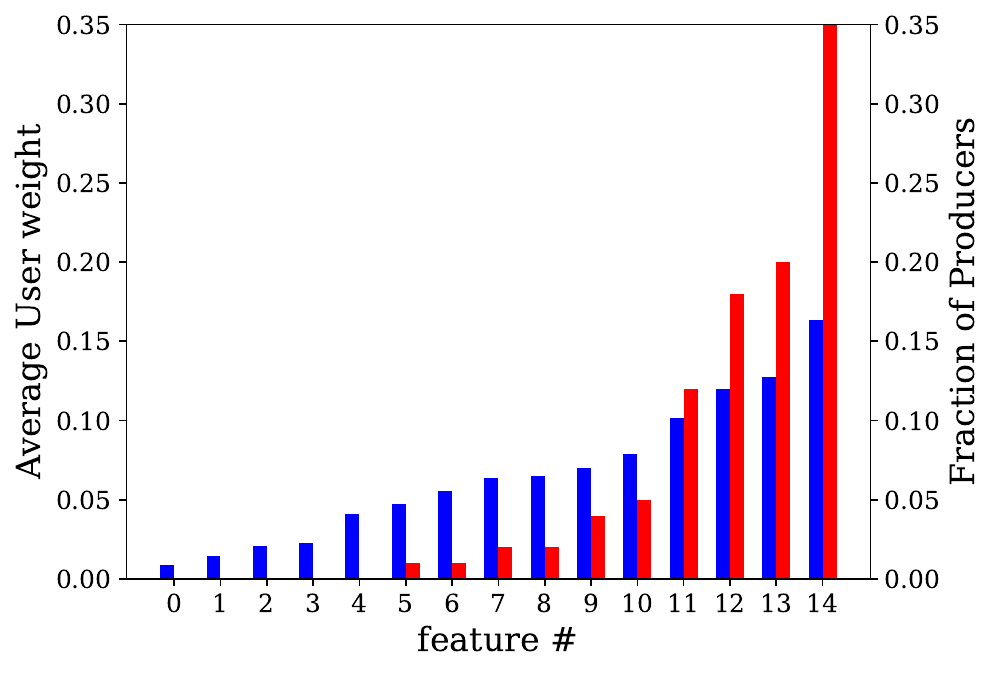}
    \label{fig:skewed-udpd-tau001}
}
\hfill
\subfloat[the linear-proportional serving rule]{
    \includegraphics[width=0.22\textwidth]{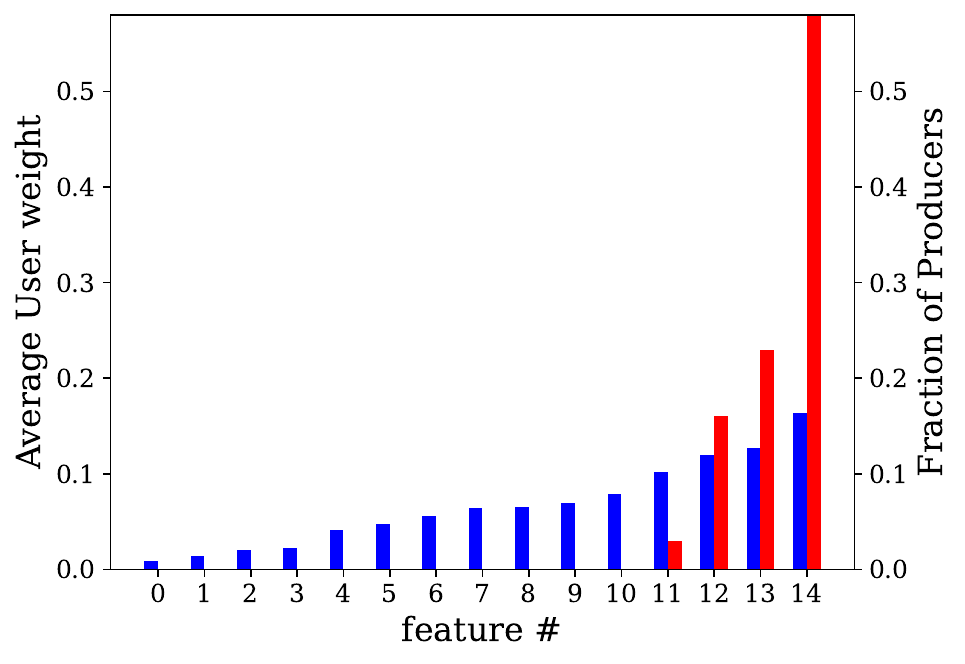}
    \label{fig:skewed-udpd-linear}
}
\caption{Average user weight on each feature (blue, left bar) and fraction of producers going for each feature (red, right bar) $n = 100$ producers, embedding dimension $d = 15$. Lower softmax temperature leads to more producer specialization. Skewed-uniform distribution of users.\label{fig:skewed-udpd}}
\end{figure}

In Figure~\ref{fig:movielens-100k-udpd}, we show the Nash equilibria for the linear-proportional serving rule and the softmax serving rule with varying temperatures on the MovieLens-100k dataset, with user embeddings obtained via NMF. For the softmax serving rule, we observe that the degree of specialization is decreasing in the temperature $\tau$. We see specialization occurring as the temperature drops from $100$ to $0.01$. A high temperature ($\tau = 100$) incentivizes homogeneous content production: this is expected, as the content-serving rule becomes largely independent of producers' decisions (content is shown with probability converging to $1/n$), and producers maximize their utility by homogeneously targeting the highest-utility content ($\argmax_{f \in d} \sum_{k \in K} c_{k}(f)$). When $\tau$ is of the order of $10$ is when we first start seeing specialization; and the lowest temperature we experiment with ($\tau=0.01$) leads to the most specialization.

In Figure~\ref{fig:skewed-udpd}, we replicate our experiment on the skewed-uniform dataset. Here too, the degree of producer specialization is decreasing in the softmax temperature. However, the levels of specialization seem to decrease for the skewed-uniform dataset and specialization seems to first occur at a lower temperature. This can be explained by the fact that low-weight features may not be worth targeting and ignored by the producers. 

Further, we note that with the linear-proportional serving rule we observe a high level of specialization on the Movielens-100k and skewed-uniform datasets as seen in Figures~\ref{fig:ml100k-udpd-linear} and~\ref{fig:skewed-udpd-linear} respectively.

In Appendix~\ref{app:full-plots-synthuniskew}, we present producer distribution figures for the uniform dataset, which exhibit insights similar to those in Figure~\ref{fig:movielens-100k-udpd}. Additionally, we provide extended figures with a few additional softmax temperatures for the Movielens-100k and Skewed datasets (see Figures~\ref{fig:fullsm-movielens-100k-udpd} and~\ref{fig:full-skewed-udpd}). Appendix~\ref{app:sparse-dataset} includes producer distribution figures for the sparse synthetic dataset.

\paragraph{Producer utility at equilibrium} In Figure \ref{fig:utility_producers-ml100k}, we plot the average producer utility with increasing softmax-serving temperature and with varying numbers of producers at a Nash Equilibrium.
We observe that the producer utility is decreasing with temperature, and temperature $\tau = 0.01$ (near-hardmax) has the highest utility. We believe this provides an argument in favour of using low temperatures in the softmax content-serving rule.  Recall that since the average user utility is identical to the average producer utility (up to a multiplicative factor), the benefit of a low temperature in the softmax-serving also extends to user utilities.

In Figure~\ref{fig:utility_serving_rules-ml100k}, we compare the average producer utility across different serving rules while varying the softmax temperature. For softmax serving, a lower temperature leads to higher producer utility, similarly, in the top-$k$ softmax rules, reducing $k$ results in higher utility at the same temperature. This indicates that ``greedier'' serving strategies improve producer utility. However, there is a trade-off between increasing producer utility by making the serving rule greedier and the potential for decreased convergence to a Nash Equilibrium, as observed in the Nash convergence table \ref{tab:nash-convergence-ml100k}. Note that the linear and round-robin serving rules are independent of temperature; we plot the same mean and standard error across all temperature values. Among these, round-robin serving yields the lowest producer utility, while linear serving remains competitive.

\begin{figure}[!h]
\centering
\captionsetup[subfigure]{justification=centering}
\subfloat[Producer utility with varying producers, softmax serving]{
    \includegraphics[width=0.4\textwidth]{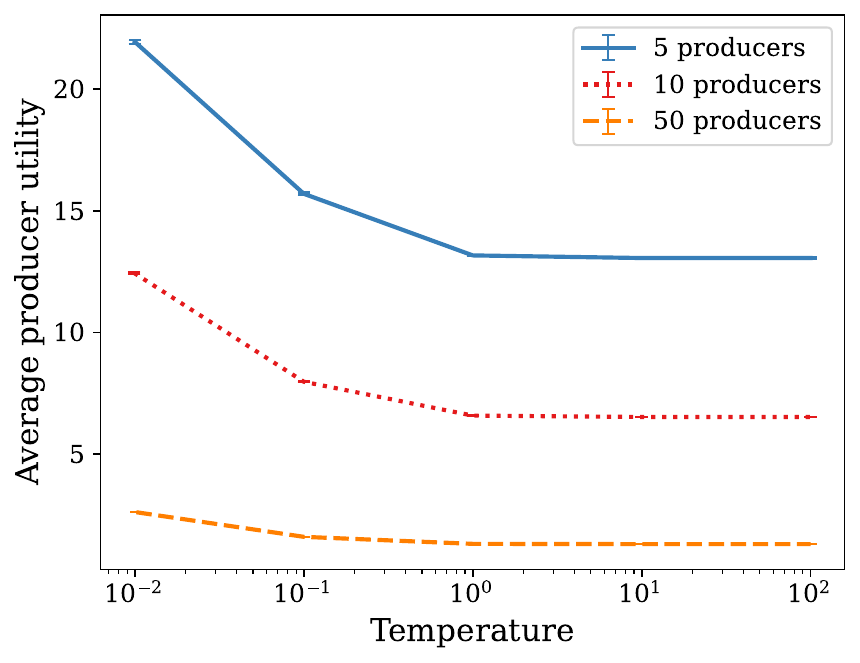}
    \label{fig:utility_producers-ml100k}
}
\hspace{0.05\textwidth}   
\subfloat[Producer utility across serving rules, $n = 50$ producers]{
    \includegraphics[width=0.4\textwidth]{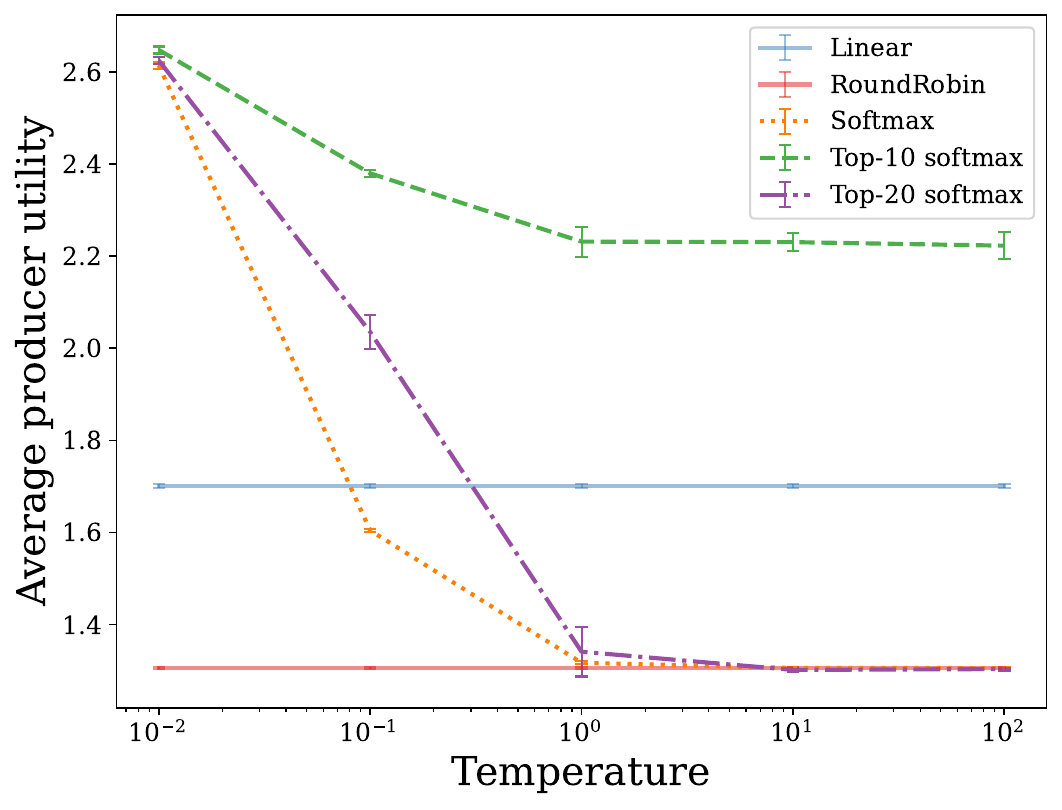}
    \label{fig:utility_serving_rules-ml100k}
}
\caption{Average producer utility on the Movielens-100k dataset. (a) Varying the number of producers $n \in \{5,10,50\}$ with embedding dimension $d=15$. (b) Comparing producer utility across serving rules: Linear (blue), RoundRobin (red), Softmax (orange), Top-10/20 Softmax (green/purple) with $n=50$ producers and $d=15$. Error bars represent standard error over 5 seeds.}
\label{fig:combined_utility_plot-ml100k}
\end{figure}

\section{Conclusion}
In this paper, we studied engagement games, a game-theoretic model of producers competing for user engagement in a recommender system. Our main structural result showed that each producer targets a single feature in embedded space at equilibrium.
We then leveraged this structural result to study content-specialization by producers and showed both theoretically and via extensive experiments that  specialization arises at a pure Nash Equilibrium of our game, and as a result of natural game dynamics. We also observe that lower temperatures in the softmax content-serving rule incentivize specialization and improve producer utility. Our linear-proportional serving rule serves as competitive benchmark still demonstrating high levels of specialization and producer utility.

\paragraph{Limitations}
As in previous work on producer competition in recommender systems, we assume that each producer is rational, and fully controls placement of their strategy $s_i$. Rationality is a common assumption in systems with strong profit motives for the producers. Full control may be less realistic, as producers can modify content features, but they do not know exactly how these changes affect the content embedding. However, this model provides a first-order approximation to real-life behavior, and the same assumptions are typical in related work \citep{hron2022modeling, meena-ss, topk-recsys, hu2023incentivizing}.

Our main structural result holds only for \emph{non-negative embeddings}.
While we make no assumption on how these non-negative embeddings are obtained, an interesting direction for future work is to study engagement games with potentially negative embeddings. There, it is not clear whether our best-response dynamics still converge, and propose to study no-regret dynamics as a direction for equilibrium computation.

Finally, we caution treating the equilibria of our engagement games as definitive; we rather present them as insights to competition in recommender systems, given the significant complexities of real-world recommender systems and environments in which they operate.

% \section*{References}
\bibliography{ref}
\bibliographystyle{plainnat}

%%%%%%%%%%%%%%%%%%%%%%%%%%%%%%%%%%%%%%%%%%%%%%%%%%%%%%%%%%%%
\newpage
\appendix

\section{Extended comparison with related work}
\label{app:relatedwork-table}
% \begin{table}[h]
% \begin{tabular}{l|l|l|l|l}
% \hline
% \textbf{Paper} & \textbf{\begin{tabular}[c]{@{}l@{}}Producer\\ Reward\end{tabular}} & \textbf{\begin{tabular}[c]{@{}l@{}}Producer\\ dynamics?\end{tabular}} & \textbf{\begin{tabular}[c]{@{}l@{}}User\\ dynamics?\end{tabular}} & \textbf{\begin{tabular}[c]{@{}l@{}}Producer \\Equilibrium type\end{tabular}}  \\\hline
% Ours  & Engagement   & Yes   &  No   & Pure NE \\ \hline
% \cite{meena-ss} & Exposure   & Yes  & No &  \\ \hline
% \cite{hron2022modeling} & Exposure   & Yes   & No   & epsilon \\ \hline
% \cite{}            & No    & Yes   & User engagement    & Dynamics \\ \hline
% \cite{prasad2023content}       & No    & Yes   & User engagement  & Dynamics      \\ \hline 
% \cite{}    & No    & Yes   & Exposure           & Equilibrium   \\ \hline
% \cite{ben2018game}             & No    & Yes   & Exposure           & Equilibrium \\ \hline
% \cite{Yao_2024_User}           & No    & Yes   & Blah  &  Dynamics    \\ \hline
% \cite{DeanPrefdyna}    & N/A   & N/A   &  Yes  &  N/A \\ \hline
% \cite{usercreatordual}  & Engagement   & Yes    &  Yes  &  N/A \\ \hline
% \end{tabular}
% \caption{Comparing our paper to related work}
% \end{table}

\begin{table}[h]
\centering
\begin{tabular}{l|l|l|l|l}
\hline
\textbf{Paper} & \textbf{\begin{tabular}[c]{@{}l@{}}Producer\\ Reward\end{tabular}} & \textbf{\begin{tabular}[c]{@{}l@{}}Producer\\ Dynamics?\end{tabular}} & \textbf{\begin{tabular}[c]{@{}l@{}}User\\ Dynamics?\end{tabular}} & \textbf{\begin{tabular}[c]{@{}l@{}}Producer \\ Equilibrium\end{tabular}} \\ \hline
Ours                          & Engagement & Yes  & No    & Pure NE       \\ \hline
\cite{ben2019convergence,ben2020content}            & Exposure         & Yes  & No        & Pure NE   \\ \hline
\cite{hron2022modeling}       & Exposure   & Yes  & No    & Local NE       \\ \hline
\cite{meena-ss}               & Exposure   & Yes  & No    & Mixed NE              \\ \hline
\cite{hu2023incentivizing}  & \begin{tabular}[c]{@{}l@{}}Mechanism design \end{tabular}         & Yes  & No            & Mixed NE      \\ \hline
\cite{topk-recsys}            & Engagement & Yes  & No    & CCE      \\ \hline
\cite{yao2023rethinking,yao2024user}  & \begin{tabular}[c]{@{}l@{}}Mechanism design \end{tabular}         & Yes  & No            & Local NE      \\ \hline
\cite{DeanPrefdyna}           & N/A        & N/A  & Yes             & N/A           \\ \hline
\cite{usercreatordual}        & Engagement & Yes  & Yes             & N/A           \\ \hline
\end{tabular}
\caption{Comparing our paper to related work}
\label{table:related_work_comparison}
\end{table}

\cite{ben2019recommendation,ben2020content} study exposure games with finite strategy spaces (producers selecting content from a finite catalog) and search for Pure NE using graph algorithms.
\cite{hron2022modeling} focus on the concept of $\varepsilon$-local Nash Equilibria: i.e., given a joint strategy profile $(s_i, s_{-i})$, each producer's $i$ strategy $s_i$ is optimal in the open ball at $(s_i, s_{-i})$ with radius $\varepsilon$---this is a weaker notion of equilibrium than pure NE, which holds globally. The work of \cite{meena-ss} characterizes the support of the mixed Nash Equilibrium of their game, but do not provide a closed-form expression for said mixed NE---we instead provide a closed-form characterization in the restricted setting of Section~\ref{sec:simple_eq}. \cite{hu2023incentivizing} builds on the exposure model in \cite{meena-ss} but studies the mechanism design problem with the platform modeled by a linear contextual bandit. \cite{topk-recsys} focus on engagement games with the top-k serving rule, studying no-regret dynamics and the Coarse Correlated Equilibria (CCE) to which these dynamics converge. \cite{yao2023rethinking, yao2024user} address the mechanism design problem, aiming to incentivize welfare-maximizing local Nash equilibria.
The works of \cite{DeanPrefdyna,usercreatordual} model user dynamics, studying user polarization and do not model producer competition and equilibrium.

% \section{Efficient heuristic for computing a pure Nash equilibrium} 

% \label{app:algorithm}

\section{Omitted Proofs} \label{app:proofs}
% \subsection{Proof of Theorem~\ref{thm:pure_eq}}\label{app:proof_theorem}
% \input{tmlr/appendix_proofs}

\subsection{Proof of Lemma~\ref{lem:proportional}}\label{app:proof_simple}
First, we note that producer $i$'s utility is convex in $s_i$ (following the same proof as in Appendix~\ref{app:proof_theorem}), however is generally not \emph{strictly} convex. This implies (as for the proof in Appendix~\ref{app:proof_theorem}) that each producer $i$ has a best response in the standard basis $\mathcal{B}$. There may however be best responses supported outside of $\mathcal{B}$, and not exist an equilibrium supported on $\mathcal{B}$. Yet, note that if there exists a strategy profile in which each producer's strategy is supported on $\mathcal{B}$ and all producers best respond to each other, it is in fact an equilibrium supported on $\mathcal{B}$: no producer can improve their utility by deviating within $\mathcal{B}$, and no producer who cannot improve by deviating on $\mathcal{B}$ can improve by deviating outside of $\mathcal{B}$ since there is a best response on $\mathcal{B}$. In turn, there exists an equilibrium supported on $\mathcal{B}$ \emph{if and only if} there exists an equilibrium when producers' strategies are restricted to $\mathcal{B}$.

Now, let us restrict ourselves to producers playing on $\mathcal{B}$. Pick any $f$ such that $\numprod_f > 0$. We start by computing the utility of producer $i$ who plays $e_f$:
\begin{align*}
u_i(e_f, s_{-i})&=\sum_{k=1}^K p_i(c_k, e_f, s_{-i}) \cdot c_k^\top e_f 
= \sum_{k=1}^K \frac{\mathbbm{1}[c_k = e_f]}{n_f}
\end{align*}
where the second equality follows from the fact that $p_i(e_f, e_f, s_{-i}) \cdot e_f^\top e_f = \frac{1}{n_f}$ and $p_i(e_{f'}, e_f, s_{-i}) = 0$ for $f' \neq f$.
Therefore, the engagement utility from user $c_k$ is $0$ if $c_k \neq e_f$, and $1/\numprod_f$ if $c_k = e_f$. Since there is a number $m_f$ of users with $c = e_f$, it follows that the total engagement utility for producer $i$, if it picks $s_i = e_f$, is given by $m_f/\numprod_f$.

Now, let us study the deviations for any given producer $i$ with $s_i = e_f$. Suppose the producer deviates to $e_{f'}$ where $f' \neq f$. Then, the new number of producers picking $f'$ is $\numprod_{f'}+1$, and producer $i$ now obtains utility $m_{f'}/(\numprod_{f'}+1)$. Therefore, producer $i$ does not deviate if and only if for all $f'$, we have 
\[
\frac{m_f}{\numprod_f} \geq \frac{m_{f'}}{\numprod_{f'} + 1}.
\]
Then $(\numprod_1,\ldots,\numprod_d)$ is an equilibrium iff no producer wants to deviate, i.e iff for all $f,f'$ such that $\numprod_f > 0$,
\[
\frac{\numprod_f}{m_f} \leq \frac{\numprod_{f'}}{m_{f'}} + \frac{1}{m_{f'}}.
\]
If $\numprod_f = 0$, there is no deviation for any producer that picks $e_f$ as such producers do not exist, and the inequality trivially holds. This concludes the proof.

\section{Supplementary results on the Movielens-100k dataset}

\subsection{Producer distribution}  
\label{app:full-plots-ml100k}
Figure \ref{fig:fullsm-movielens-100k-udpd} shows the producer distribution for the Movielens-100k dataset with softmax serving, illustrating increased specialization as the temperature decreases. Similarly, Figure \ref{fig:top20-movielens-100k-udpd} shows similar insights for top-20 softmax serving, however, Unlike full softmax, top-20 serving shows specialization even at high temperatures ($\tau = 100$), as it retains the top 20 producers and serves those producers almost randomly (high $\tau$). Figure \ref{fig:ml100k-nontemp-udpd} presents producer distributions for the three non-temperature-dependent rules: greedy, linear, and round-robin. Greedy and linear serving result in producer specialization, while round-robin, as highlighted in remark \ref{rem:round-robin-conv}), leads all producers to target the highest-weight feature.

\begin{figure}[!h]
\centering
\captionsetup[subfigure]{justification=centering}
\subfloat[Softmax $\tau = 100$]{
    \includegraphics[width=0.25\textwidth]{NewFigures/udpd/movielens/movielens100k_temp100_dim_15.pdf}
\label{fig:movielens-100k-udpd-tau100}
}
\hspace{0.05\textwidth}   
\subfloat[Softmax $\tau = 10$]{
    \includegraphics[width=0.25\textwidth]{NewFigures/udpd/movielens/movielens100k_temp10_dim_15.pdf}}
\hspace{0.05\textwidth}
\subfloat[Softmax $\tau = 1$]{
    \includegraphics[width=0.25\textwidth]{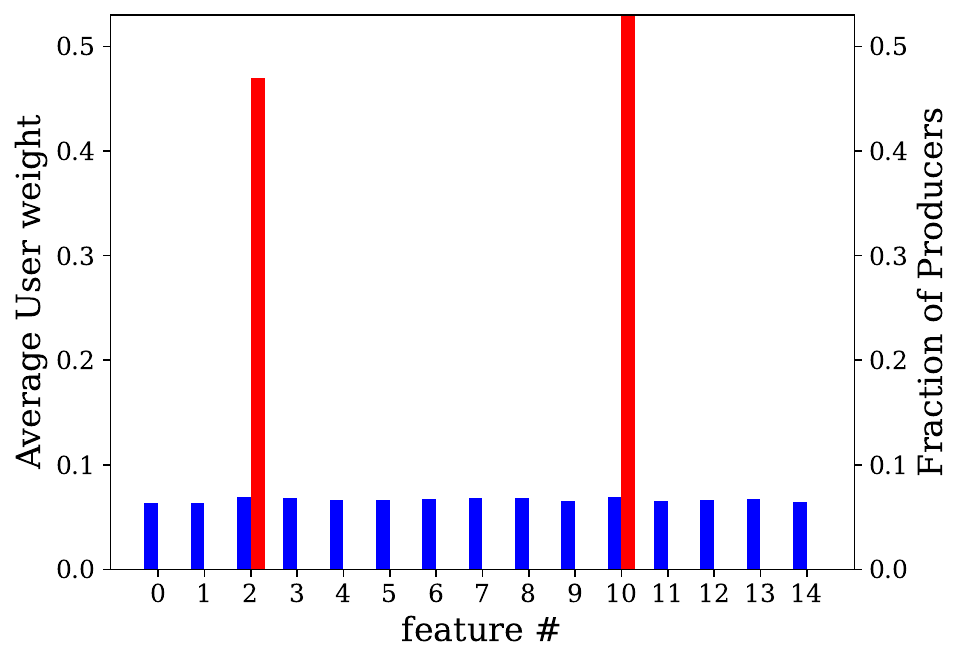}
}
\hspace{0.05\textwidth}
\\
\subfloat[Softmax $\tau = 0.1$]{
    \includegraphics[width=0.25\textwidth]{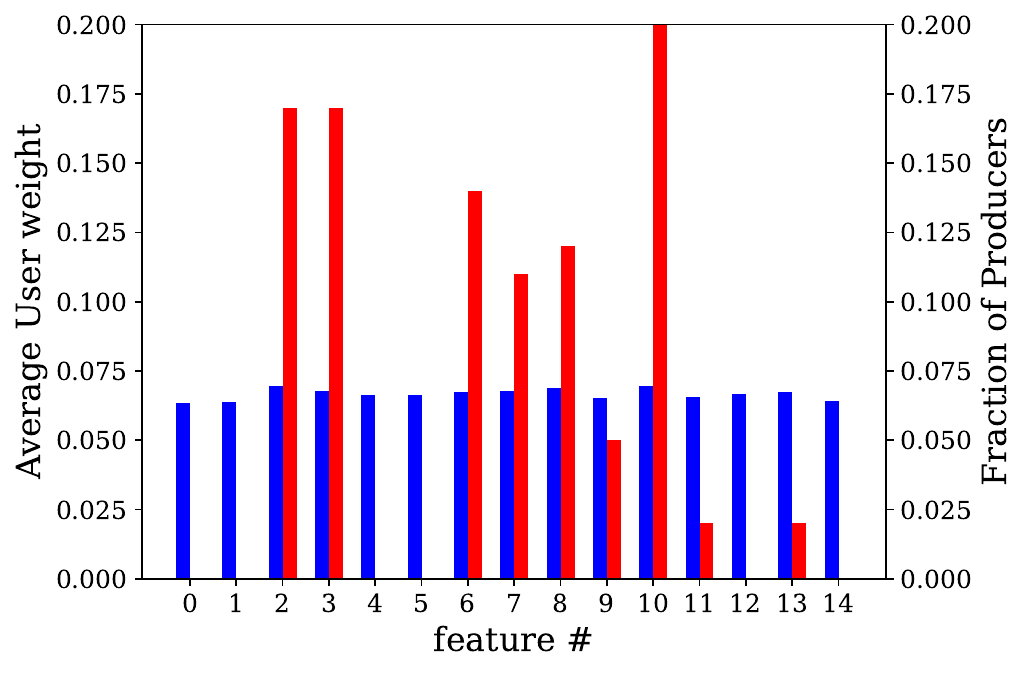}
}
\hspace{0.05\textwidth}
\subfloat[Softmax $\tau = 0.01$]{
    \includegraphics[width=0.25\textwidth]{NewFigures/udpd/movielens/movielens100k_temp001_dim_15.pdf}
\label{fig:ml100k-udpd-temp001-full}
}
\caption{(Full) Softmax serving: Average user weight on each feature (blue, left bar) and fraction of producers going for each feature (red, right bar) $n = 100$ producers, embedding dimension $d = 15$. User embeddings obtained from NMF on MovieLens-100k.}
\label{fig:fullsm-movielens-100k-udpd}
\end{figure}

\begin{figure}[!h]
\centering
\captionsetup[subfigure]{justification=centering}
\subfloat[Softmax $\tau = 100$]{
    \includegraphics[width=0.25\textwidth]{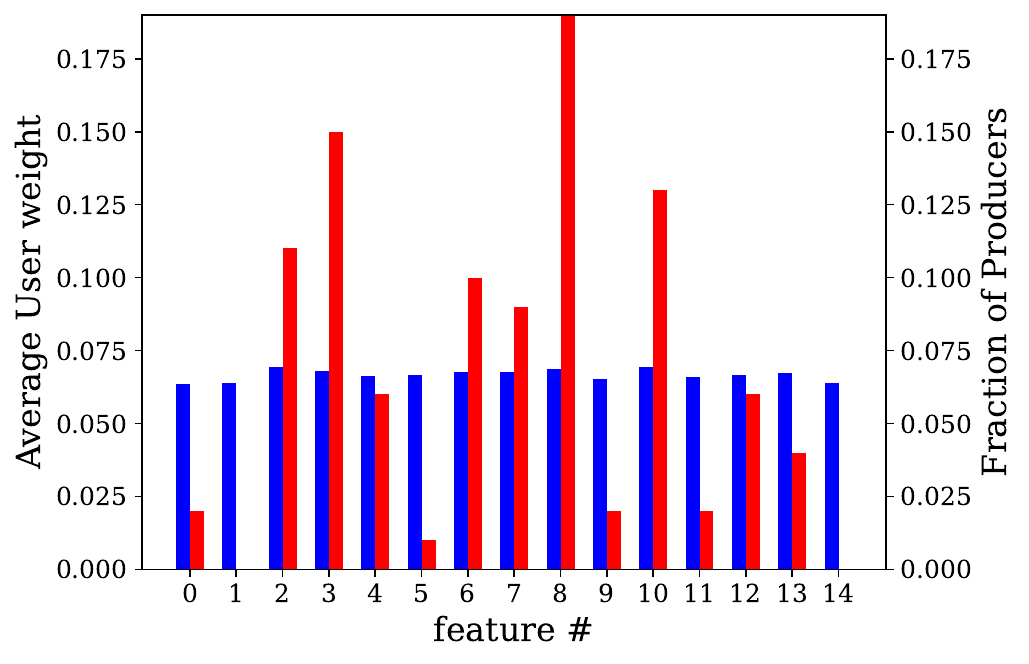}
}
\hspace{0.05\textwidth}   
\subfloat[Softmax $\tau = 10$]{
    \includegraphics[width=0.25\textwidth]{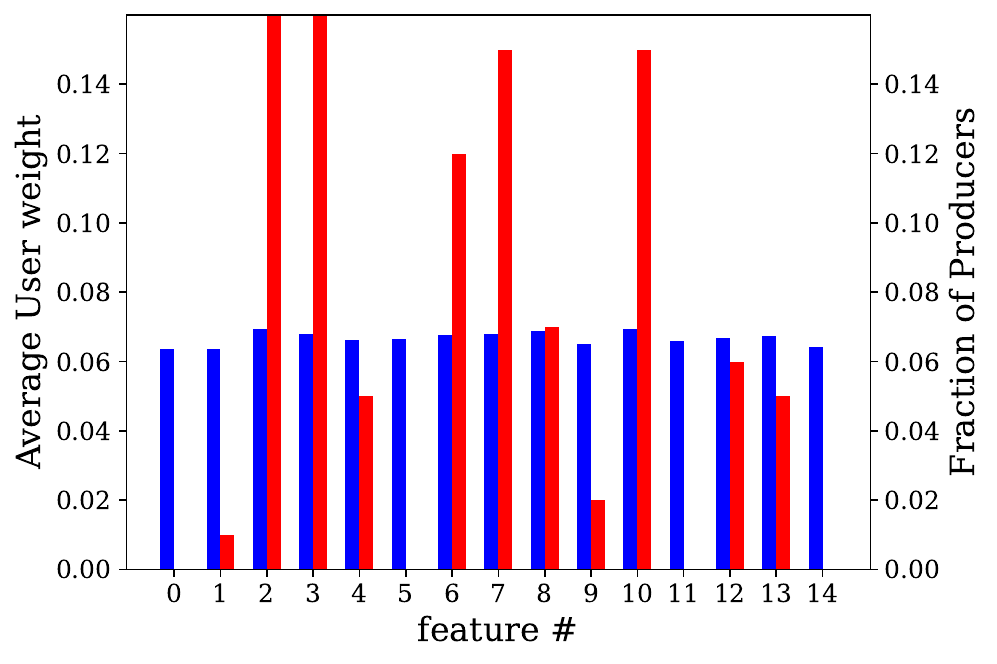}}
\hspace{0.05\textwidth}
\subfloat[Softmax $\tau = 1$]{
    \includegraphics[width=0.25\textwidth]{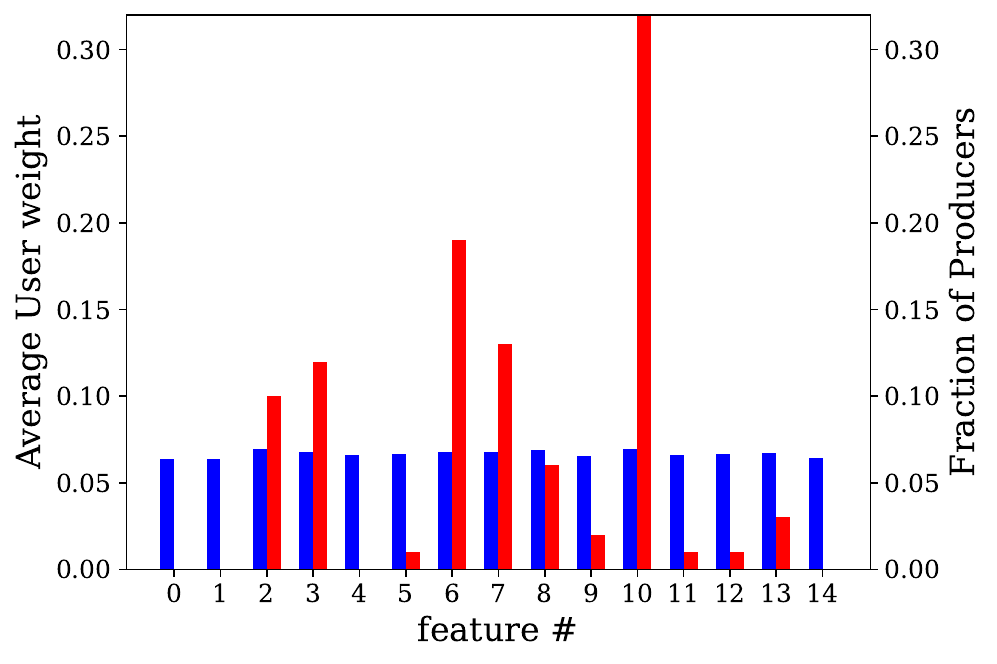}
}
\hspace{0.05\textwidth}
\\
\subfloat[Softmax $\tau = 0.1$]{
    \includegraphics[width=0.25\textwidth]{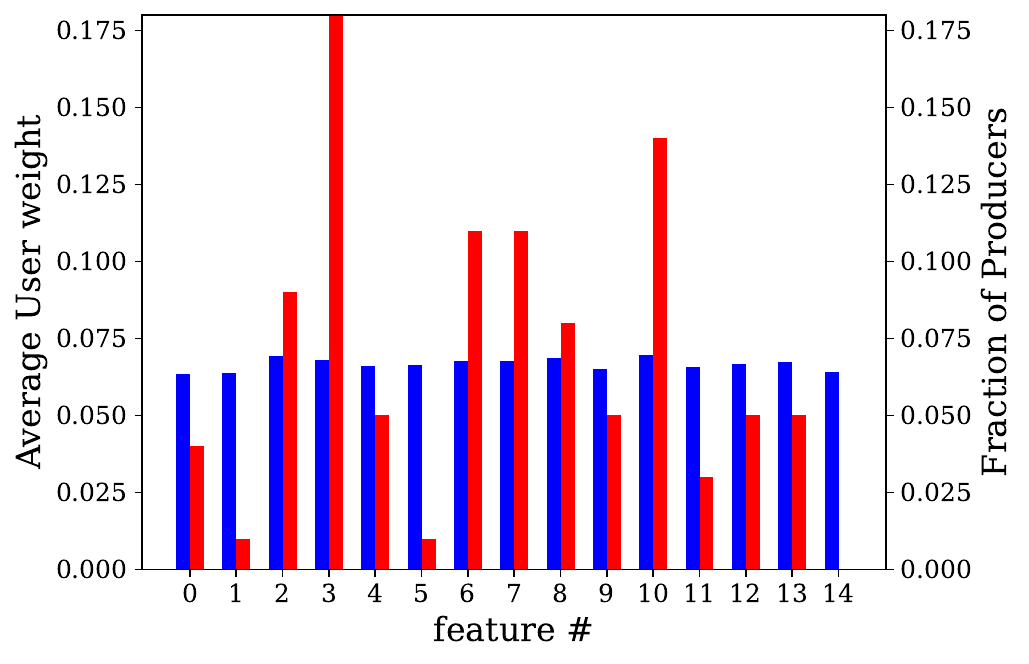}
}
\hspace{0.05\textwidth}
\subfloat[Softmax $\tau = 0.01$]{
    \includegraphics[width=0.25\textwidth]{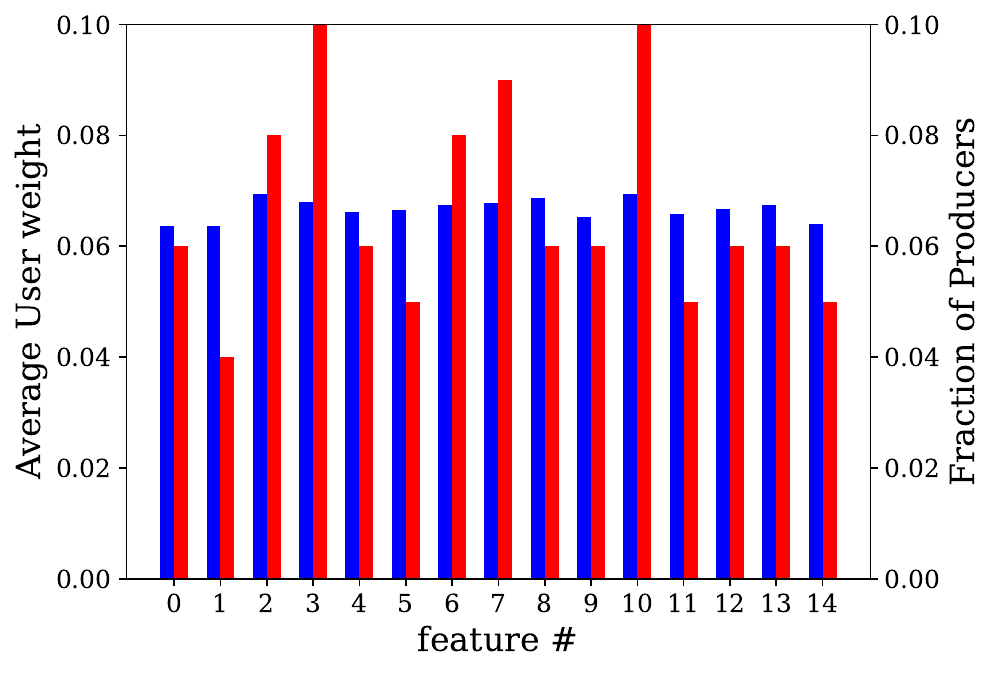}
}
\caption{(Top-20) Softmax serving: Average user weight on each feature (blue, left bar) and fraction of producers going for each feature (red, right bar) $n = 100$ producers, embedding dimension $d = 15$. User embeddings obtained from NMF on MovieLens-100k, Note: unlike full softmax, top-20 serving shows specialization even at a high temperature.\label{fig:top20-movielens-100k-udpd}}
\end{figure}

\begin{figure}[!h]
\centering
\captionsetup[subfigure]{justification=centering}
\subfloat[Greedy]{
    \includegraphics[width=0.25\textwidth]{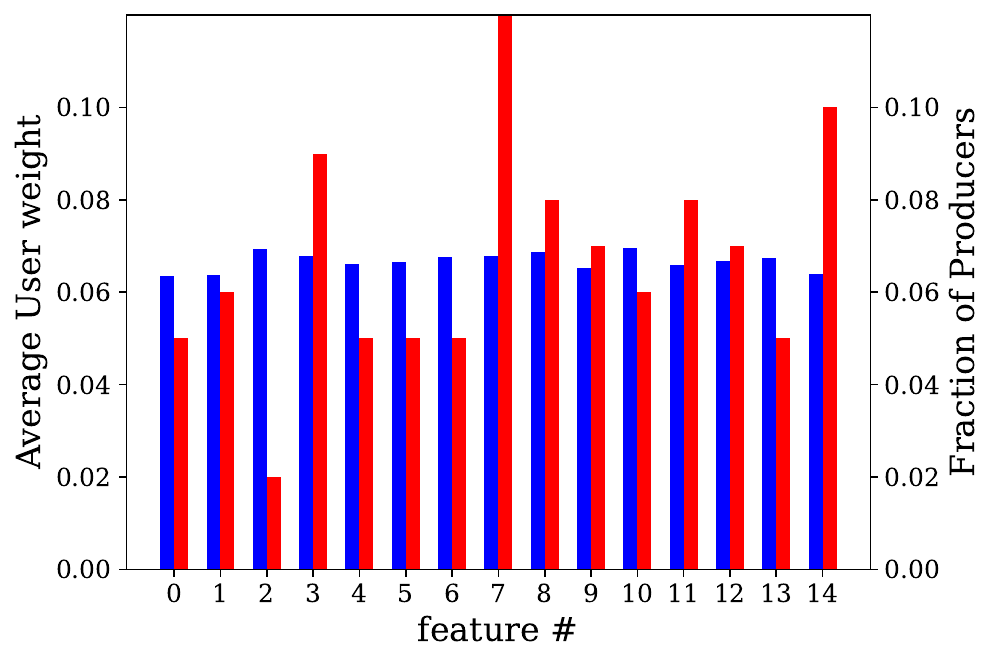}
}
\hspace{0.05\textwidth}   
\subfloat[Linear]{
    \includegraphics[width=0.25\textwidth]{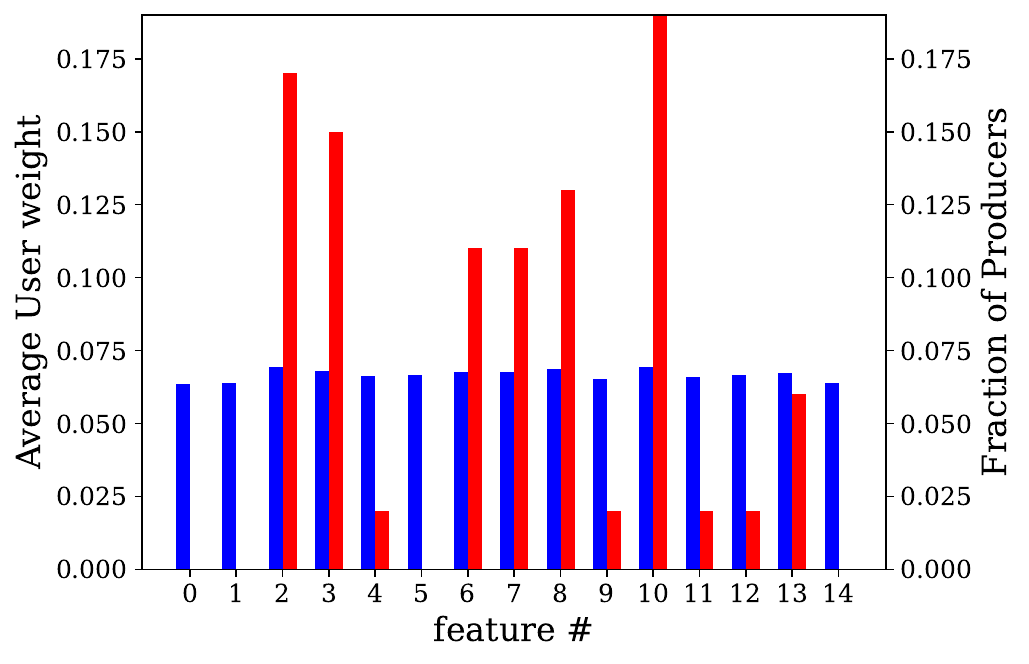}}
\hspace{0.05\textwidth}
\subfloat[Round robin]{
    \includegraphics[width=0.25\textwidth]{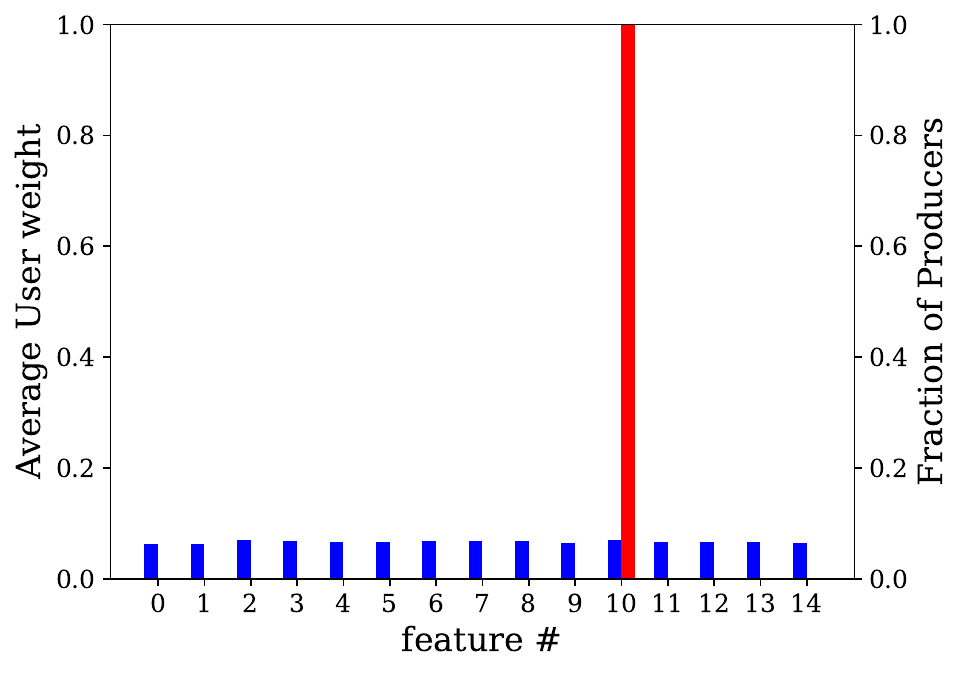}
}
\caption{Temperature independent serving rules: Greedy, Linear and Round-robin serving, Average user weight on each feature (blue, left bar) and fraction of producers going for each feature (red, right bar) $n = 100$ producers, embedding dimension $d = 15$.User embeddings obtained from NMF on MovieLens-100k. \label{fig:ml100k-nontemp-udpd}}
\end{figure}

\newpage

\section{Supplementary results on the uniform and skewed synthetic datasets}

\subsection{Number of iterations until convergence} \label{app:numiters-mbody}
In Figure \ref{fig:num_iters_convergence-synthetic}, we plot the number of iterations  (averaged over $40$ runs) of Algorithm~\ref{alg:bestrep_dynamics}  with increasing number of producers and across varying embedding dimensions. We do so on the Uniform and Skewed-uniform datasets. With both the linear-proportional and softmax content serving rules, we observe that Algorithm~\ref{alg:bestrep_dynamics} seems to scale linearly in the number of producers, further highlighting the computational efficiency of our heuristic.

\begin{figure}[!h]
\centering
\captionsetup[subfigure]{justification=centering}
\subfloat[the linear-proportional serving rule uniform dataset]{
    \includegraphics[width=0.23\textwidth]{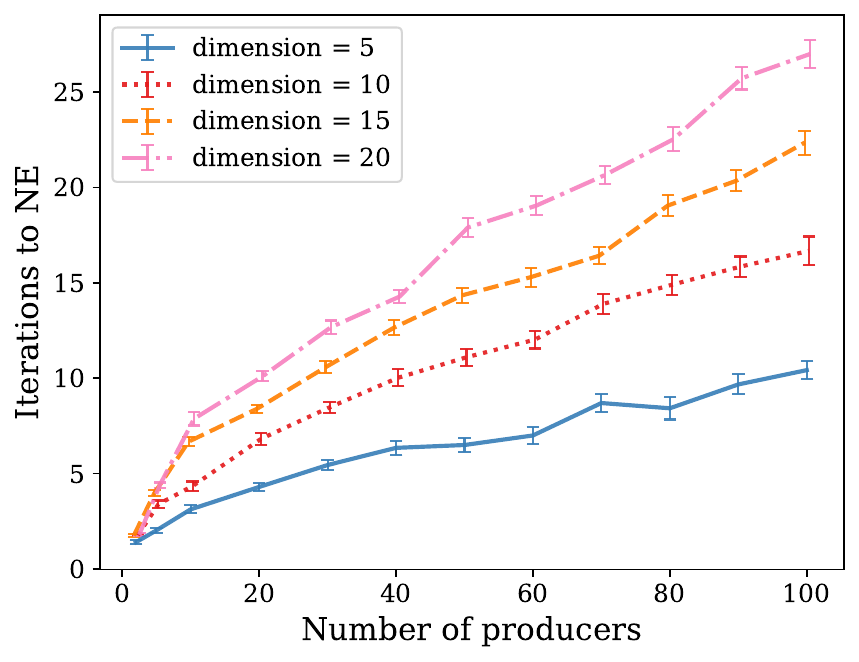}
    \label{fig:convlinear-uniform}
}
\subfloat[the linear-proportional serving rule skewed-uniform dataset]{
    \includegraphics[width=0.23\textwidth]{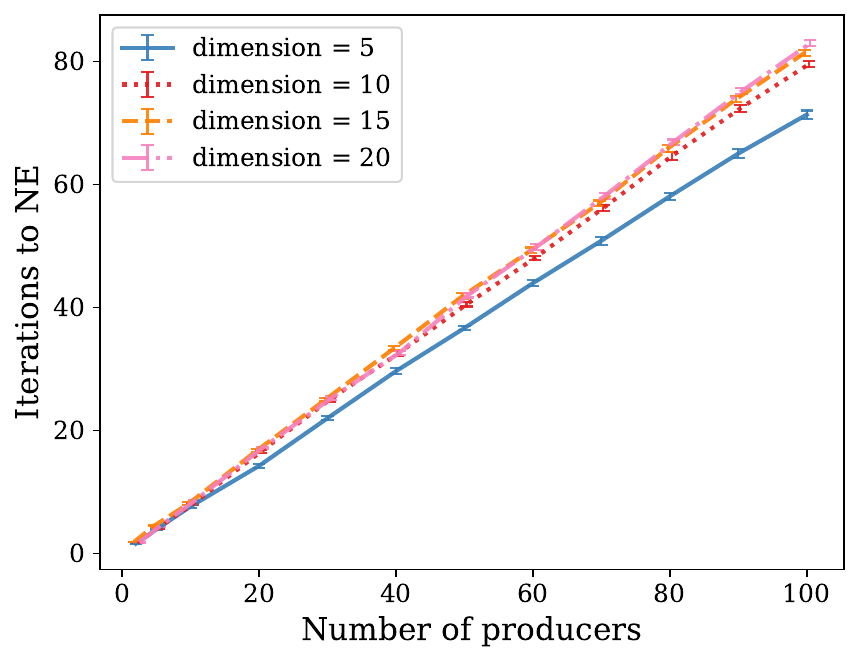}
    \label{fig:convlinear-skewed}
}
\subfloat[the softmax serving rule uniform dataset]{
    \includegraphics[width=0.23\textwidth]{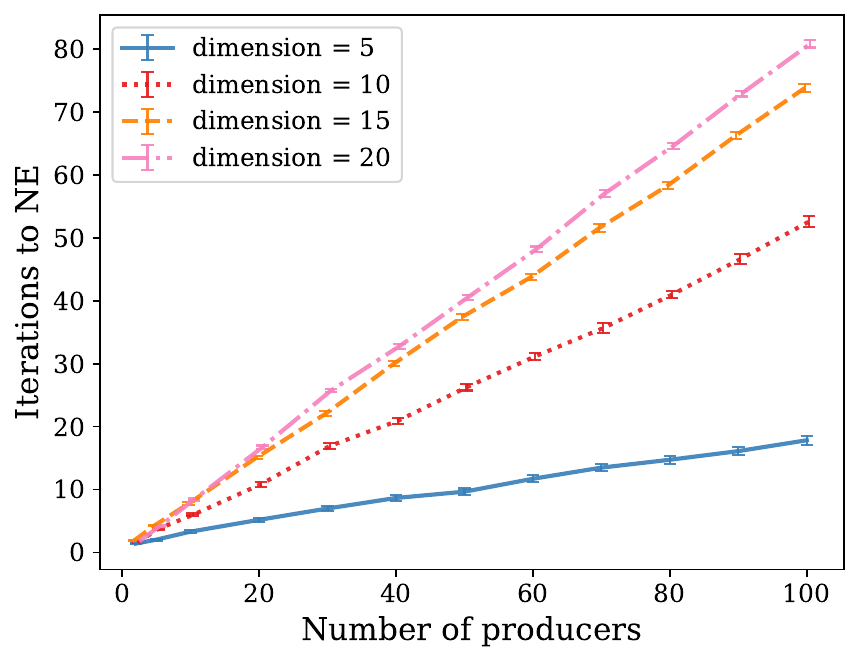}\label{fig:convsoftmax-uniform}
}
\subfloat[the softmax serving rule skewed-uniform dataset]{
    \includegraphics[width=0.23\textwidth]{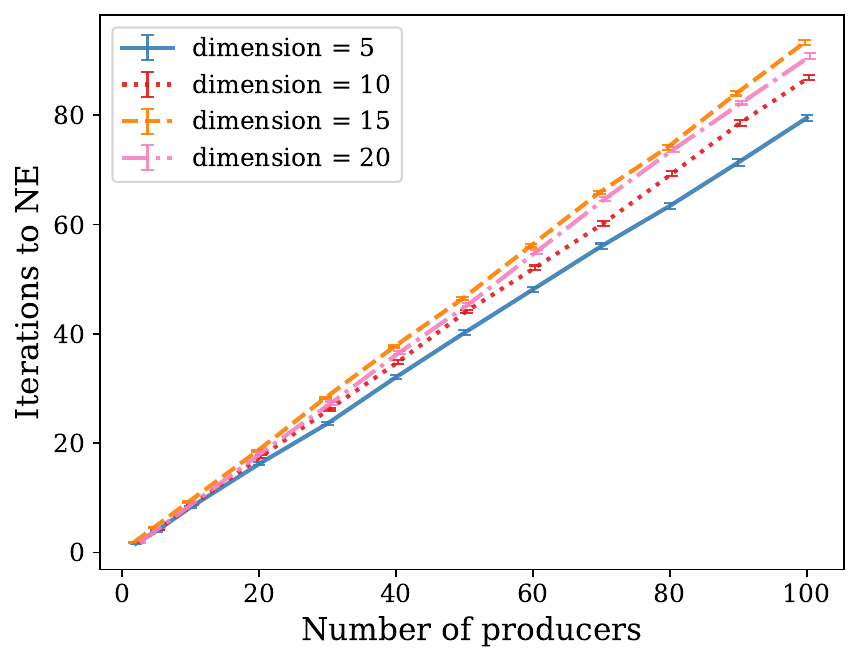}
    \label{fig:convsoftmax-skewed}
} 
\caption{Number of iterations of Algorithm~\ref{alg:bestrep_dynamics} until convergence to a Nash Equilibrium on the uniform, skewed-uniform datasets. The different curves represent different embedding dimensions in the game $d \in \{5,10,15,20\}$; the error bars represent standard error over $40$ runs.}\label{fig:num_iters_convergence-synthetic}
\end{figure}

\subsection{Producer distribution}  
\label{app:full-plots-synthuniskew}
Figure~\ref{fig:full-uniform_udpd} and Figures \ref{fig:full-skewed-udpd} provide producer distribution plots for softmax serving on the synthetic uniform and skewed datasets. These plots further highlight how the degree of specialization increases as the temperature decreases, over a few more softmax temperatures ($\tau = 10$ and $\tau = 0.1$). 

\begin{figure}[!h]
\centering
\captionsetup[subfigure]{justification=centering}
\subfloat[Softmax $\tau = 100$]{
    \includegraphics[width=0.23\textwidth]{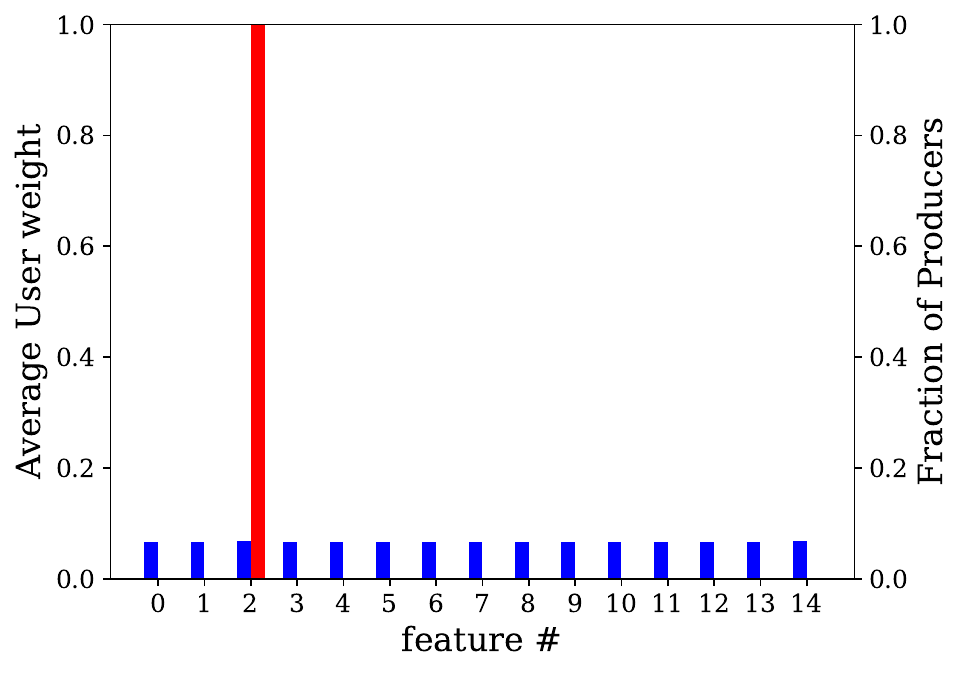}
    \label{fig:uniform-udpd-temp100-full}
}
\hspace{0.05\textwidth}
\subfloat[Softmax $\tau = 10$]{
    \includegraphics[width=0.23\textwidth]{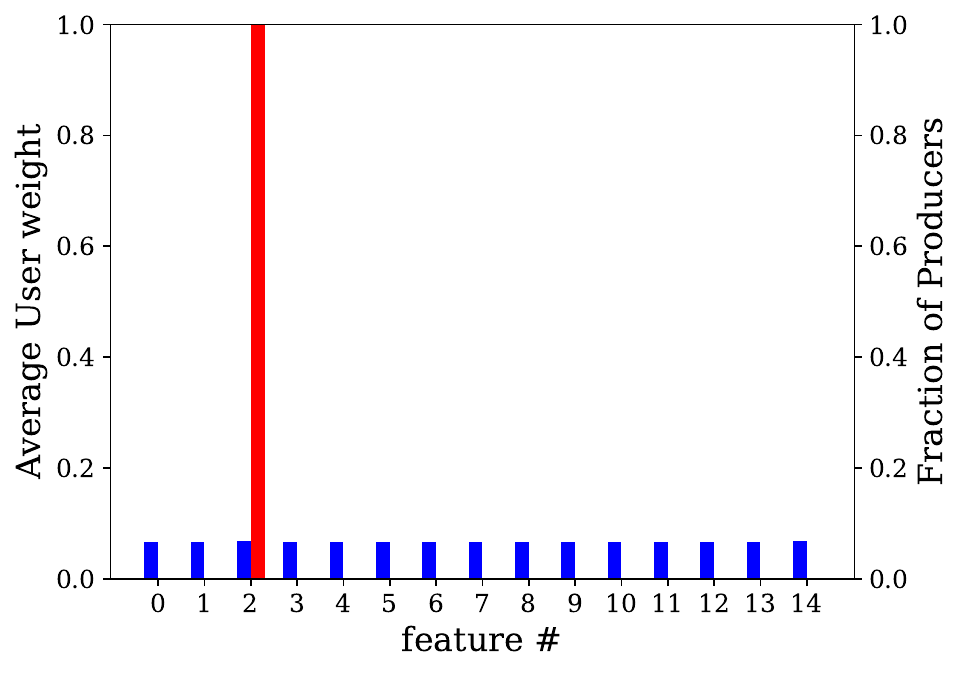}
    \label{fig:uniform-udpd-temp10-full}
}
\hspace{0.05\textwidth}
\subfloat[Softmax $\tau = 1$]{
    \includegraphics[width=0.23\textwidth]{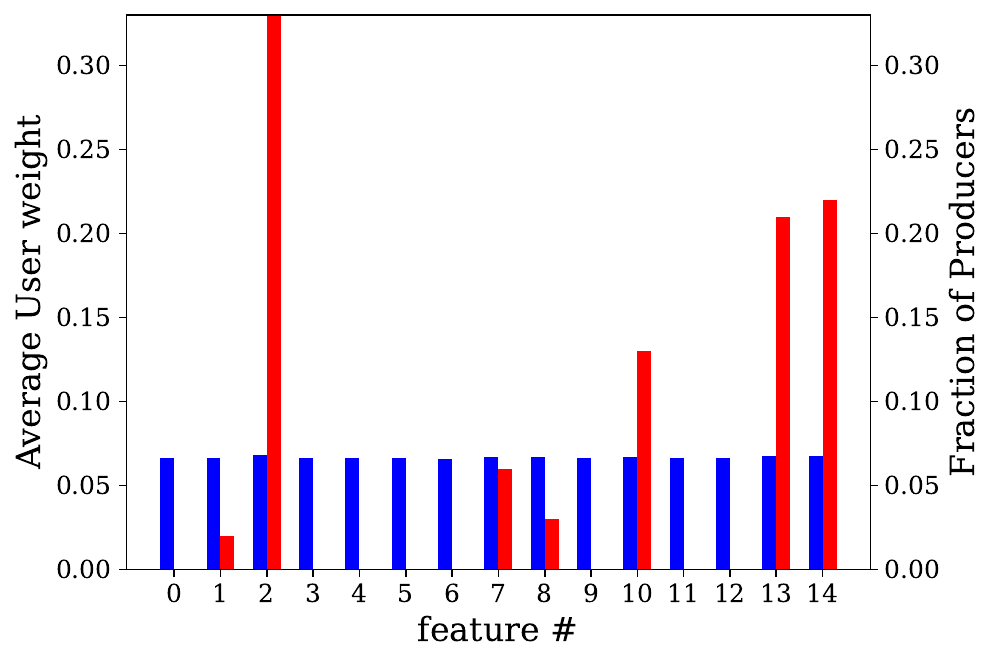}
    \label{fig:uniform-udpd-temp1-full}
}
\\
\subfloat[Softmax $\tau = 0.1$]{
    \includegraphics[width=0.23\textwidth]{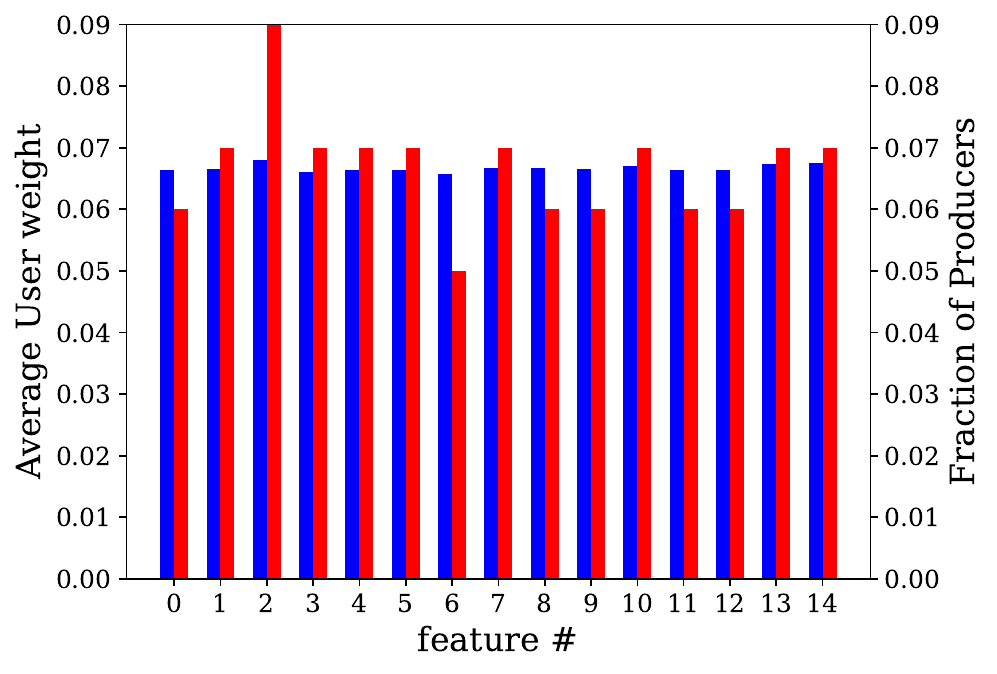}
    \label{fig:uniform-udpd-temp01-full}
}
\hspace{0.05\textwidth}
\subfloat[Softmax $\tau = 0.01$]{
    \includegraphics[width=0.23\textwidth]{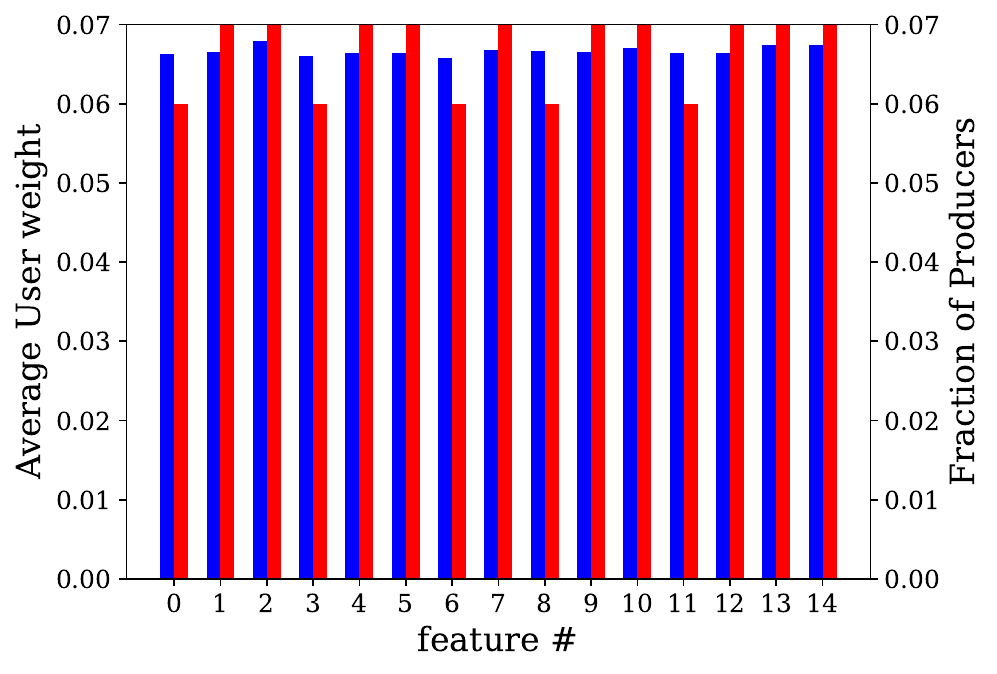}
    \label{fig:uniform-udpd-temp001-full}
}
\hspace{0.05\textwidth}
\subfloat[Linear-proportional serving rule]{
    \includegraphics[width=0.23\textwidth]{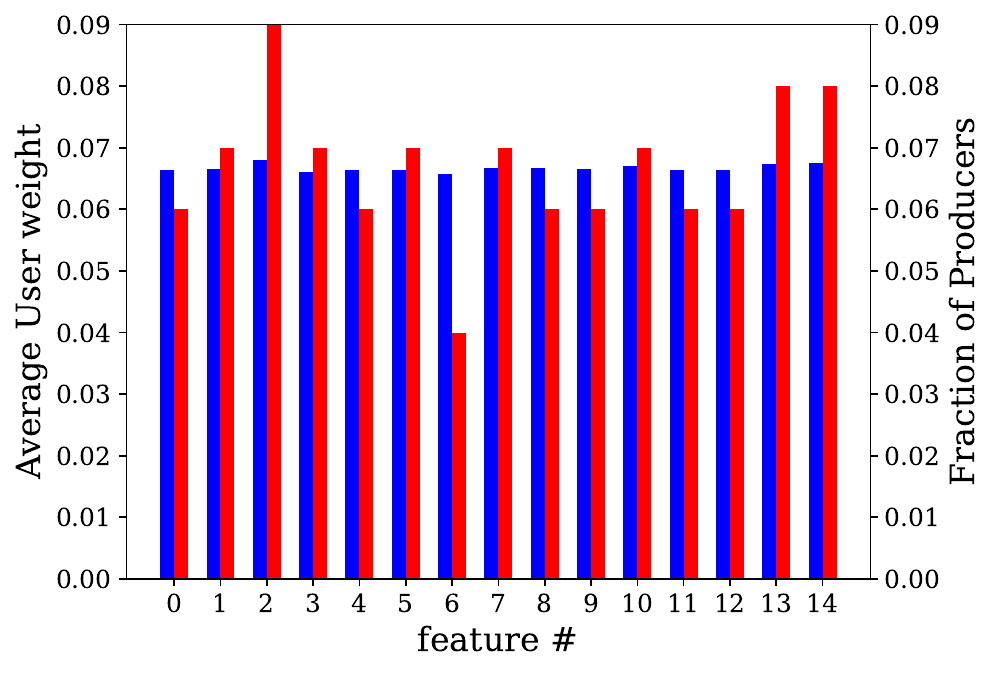}
        \label{fig:uniform-udpd-linear-full}
}
\caption{Average user weight on each feature (blue, left bar) and fraction of producers going for each feature (red, right bar) $n = 100$ producers, embedding dimension $d = 15$. Lower softmax temperature leads to more producer specialization. Uniform distribution of users. \label{fig:full-uniform_udpd}}
\end{figure}

\begin{figure}[H]
\centering
\captionsetup[subfigure]{justification=centering}
\subfloat[Softmax $\tau = 100$]{
    \includegraphics[width=0.25\textwidth]{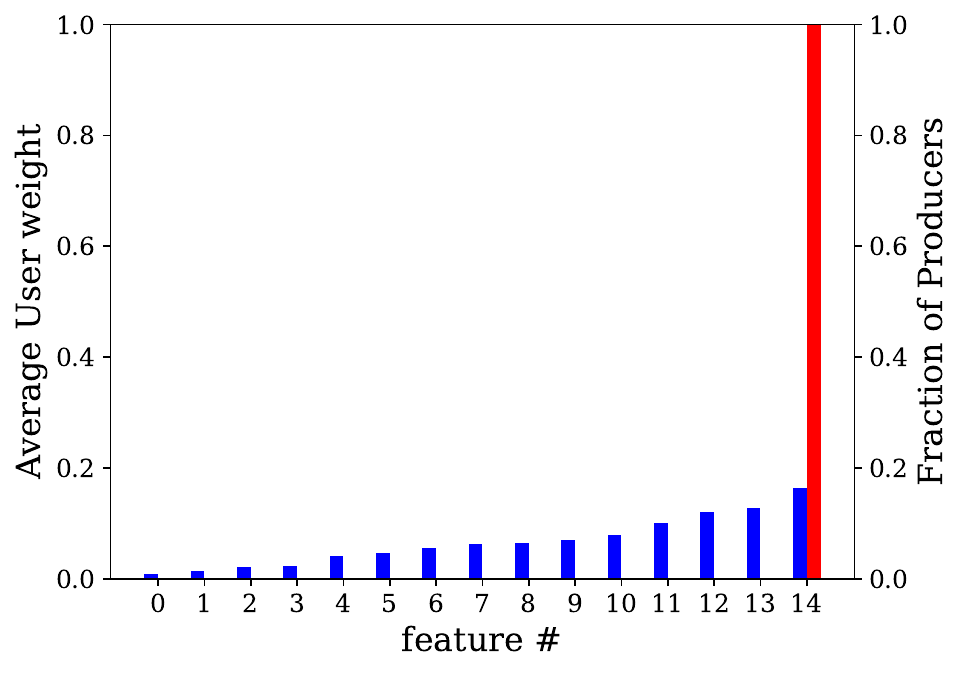}
    \label{fig:skewed-udpd-tau100-full}
}
\hspace{0.05\textwidth}   
\subfloat[Softmax $\tau = 10$]{
    \includegraphics[width=0.25\textwidth]{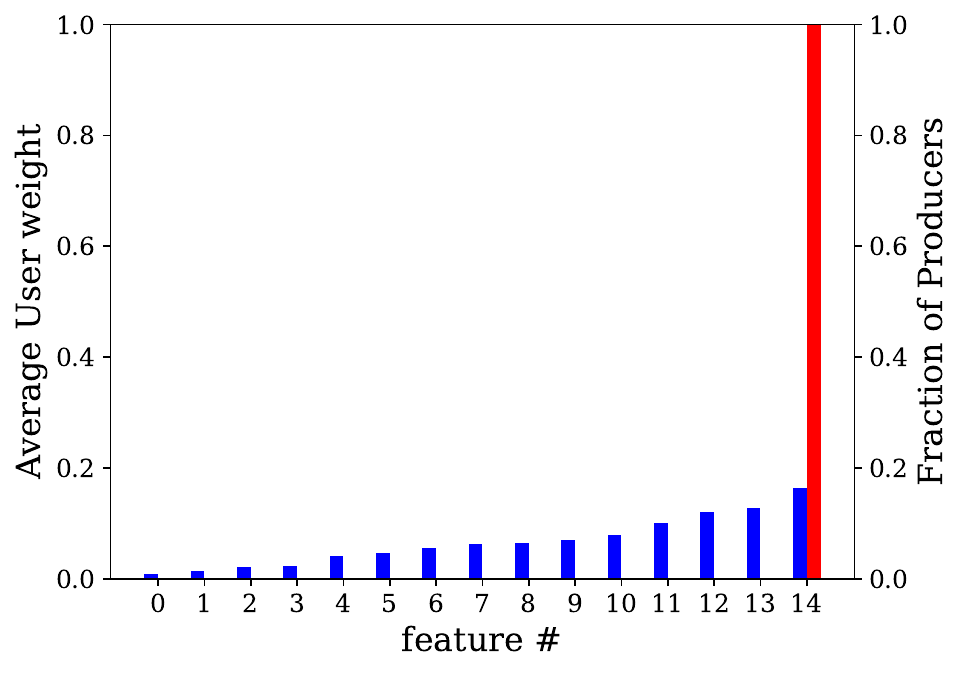}}
\hspace{0.05\textwidth}   
\subfloat[Softmax $\tau = 1$]{
    \includegraphics[width=0.25\textwidth]{NewFigures/udpd/synth-skewed/synth-skewed_temp1_dim_15.pdf}
}
\\
\subfloat[Softmax $\tau = 0.1$]{
    \includegraphics[width=0.25\textwidth]{NewFigures/udpd/synth-skewed/synth-skewed_temp01_dim_15.pdf}
}
\hspace{0.05\textwidth}   
\subfloat[Softmax $\tau = 0.01$]{
    \includegraphics[width=0.25\textwidth]{NewFigures/udpd/synth-skewed/synth-skewed_temp001_dim_15.pdf}
    \label{fig:skewed-udpd-tau001-full}
}
\hspace{0.05\textwidth}   
\subfloat[the linear-proportional serving rule]{
    \includegraphics[width=0.25\textwidth]{NewFigures/udpd/synth-skewed/synth-skewed_linear_dim_15.pdf}
    \label{fig:skewed-udpd-linear-full}
}
\caption{Average user weight on each feature (blue, left bar) and fraction of producers going for each feature (red, right bar) $n = 100$ producers, embedding dimension $d = 15$. Lower softmax temperature leads to more producer specialization. Skewed-uniform distribution of users.
}
\label{fig:full-skewed-udpd}
\end{figure}

\newpage

\section{{Experiments on the Sparse synthetic dataset}}
\label{app:sparse-dataset}
\subsection{Average producer utility}
Figure~\ref{fig:spsuni-allservingrules-util} illustrates the average producer utility across serving rules for the sparse synthetic dataset, revealing trends similar to those observed in Figure~\ref{fig:utility_serving_rules-ml100k} in the main text. Specifically, lower $k$ values and lower temperatures lead to higher producer utility; however, this comes with a decreased likelihood of convergence to a Nash equilibrium.

\begin{figure}[H]
\centering
\includegraphics[width=0.4\textwidth]{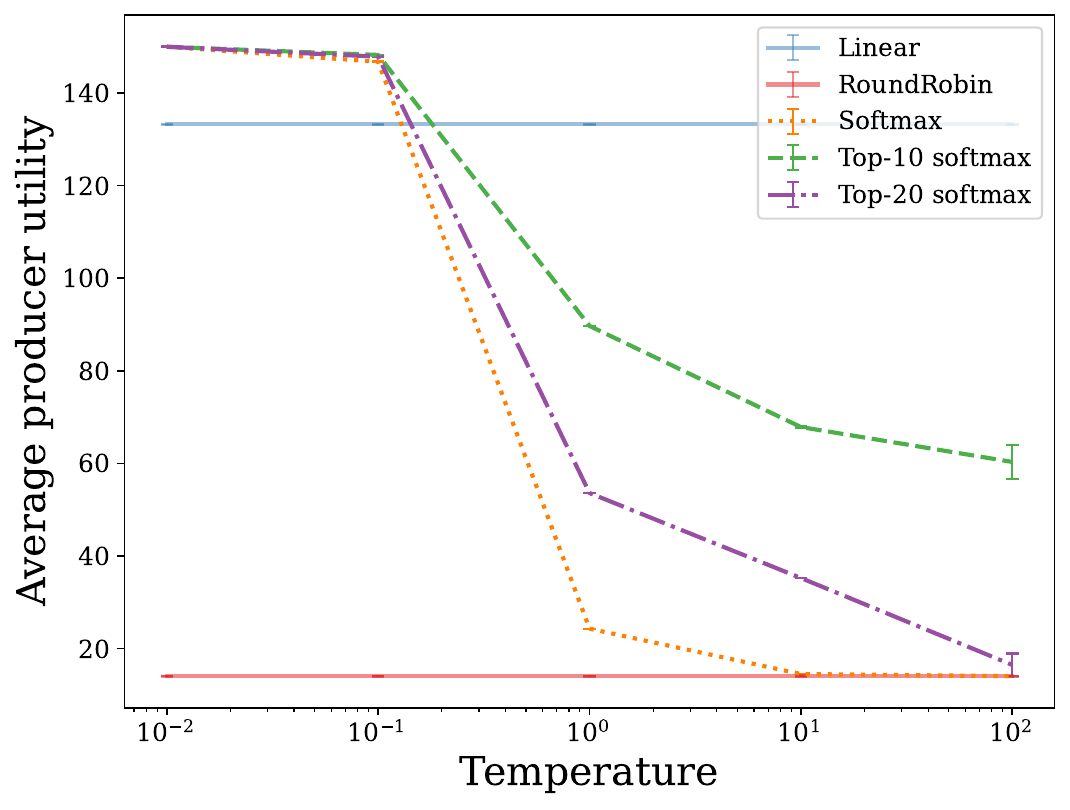}

\caption{Average producer utility on the Sparse Synthetic dataset: comparing producer utility across serving rules: Linear (blue), RoundRobin (red), Softmax (orange), Top-10/20 Softmax (green/purple) with $n=50$ producers and $d=15$. Error bars represent standard error over 5 seeds.\label{fig:spsuni-allservingrules-util}}
\end{figure}

\subsection{Producer distribution at Nash Equilibrium}

\begin{figure}[H]
\centering
\captionsetup[subfigure]{justification=centering}
\subfloat[Softmax $\tau = 100$]{
    \includegraphics[width=0.25\textwidth]{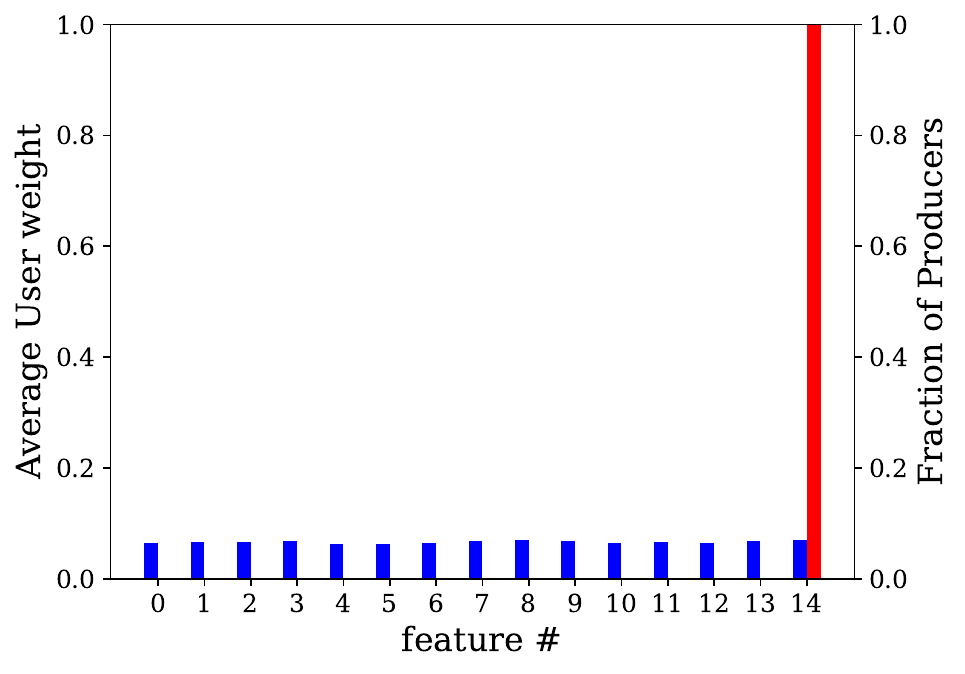}
\label{fig:spsuni-udpd-tau100}
}
\hspace{0.05\textwidth}   
\subfloat[Softmax $\tau = 10$]{
    \includegraphics[width=0.25\textwidth]{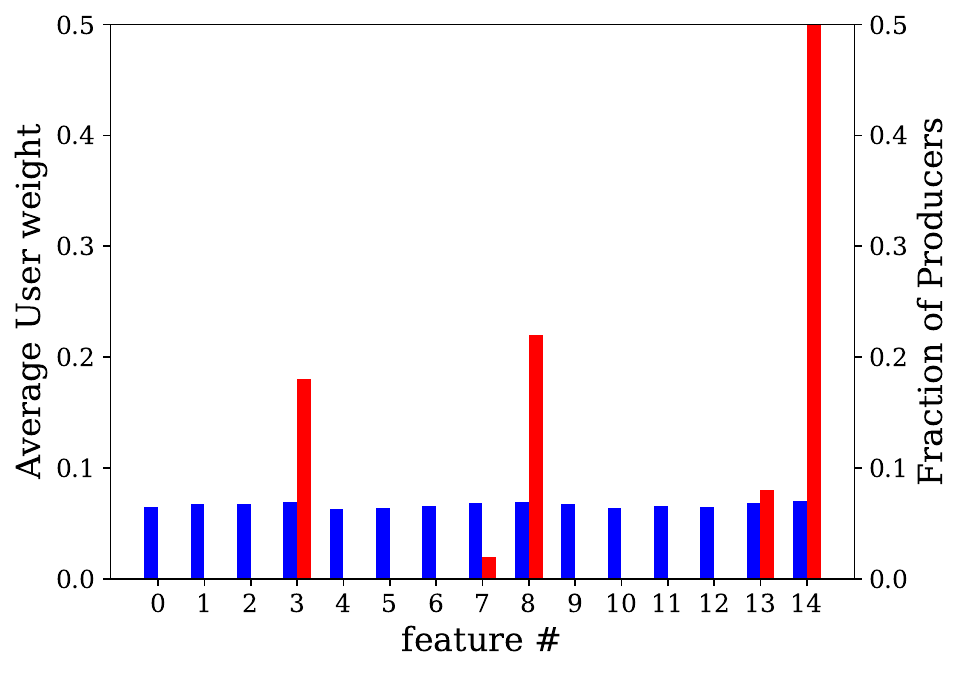}}
\hspace{0.05\textwidth}
\subfloat[Softmax $\tau = 1$]{
    \includegraphics[width=0.25\textwidth]{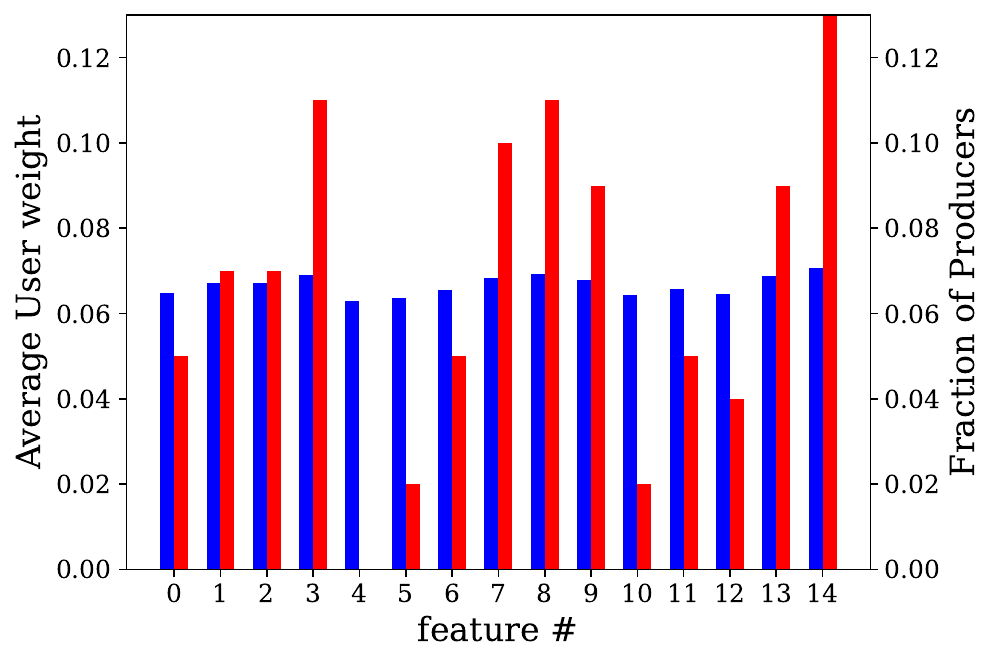}
}
\hspace{0.05\textwidth}
\\
\subfloat[Softmax $\tau = 0.1$]{
    \includegraphics[width=0.25\textwidth]{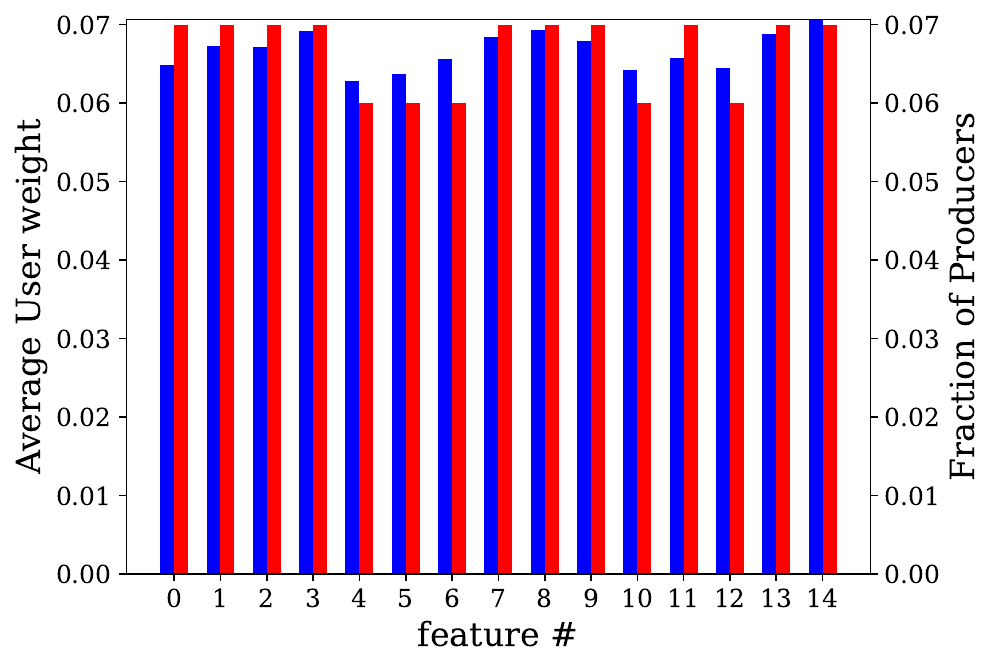}
}
\hspace{0.05\textwidth}
\subfloat[Softmax $\tau = 0.01$]{
    \includegraphics[width=0.25\textwidth]{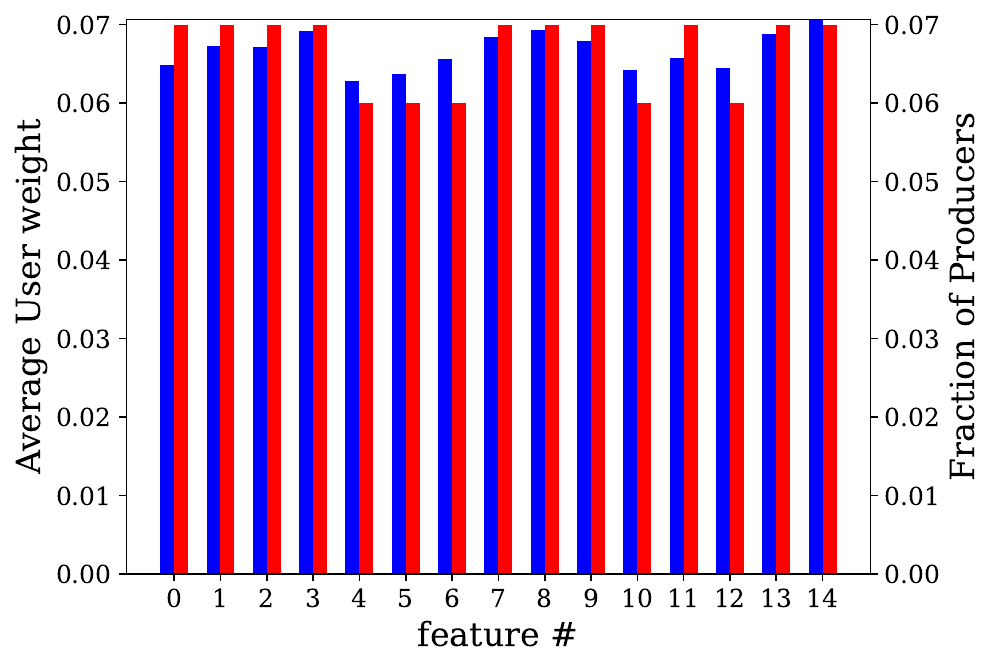}
\label{fig:spsuni-udpd-temp001-full}
}
\hspace{0.05\textwidth}
\subfloat[the linear-proportional serving rule]{
    \includegraphics[width=0.25\textwidth]{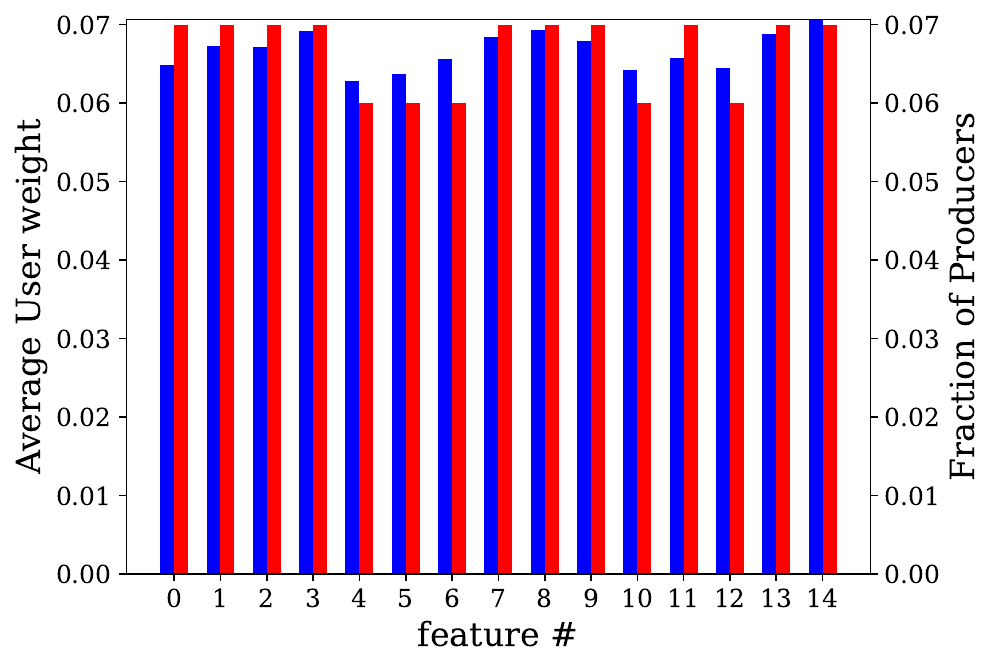}
\label{fig:spsuni-udpd-linear-full}
}
\caption{Average user weight on each feature (blue, left bar) and fraction of producers going for each feature (red, right bar) $n = 100$ producers, embedding dimension $d = 15$. Lower softmax temperature leads to more producer specialization. Sparse distribution of users.
\label{fig:full-spsuni-udpd}}
\end{figure}

% \section{Non convergence analysis}
% \input{tmlr/non-convergence-investigation}

\section{Experiments on the AmazonMusic and RentTheRunway datasets}\label{app:experiments}

\subsection{Number of iterations till convergence}\label{app:num_iters}
Figure \ref{fig:num_iters_convergence_amzn-rtr} plots the number of iterations of Algorithm \ref{alg:bestrep_dynamics} until convergence to NE, averaged over $40$ runs of best response dynamics, on the AmazonMusic and RentTheRunway datasets. Similar to Figure \ref{fig:num_iters_convergence}, we note a fast time to convergence that seems to scale linearly in the number of producers.

\begin{figure}[!h]
\centering
\subfloat[\centering Linear serving AmazonMusic]{
    \includegraphics[width=0.23\textwidth]{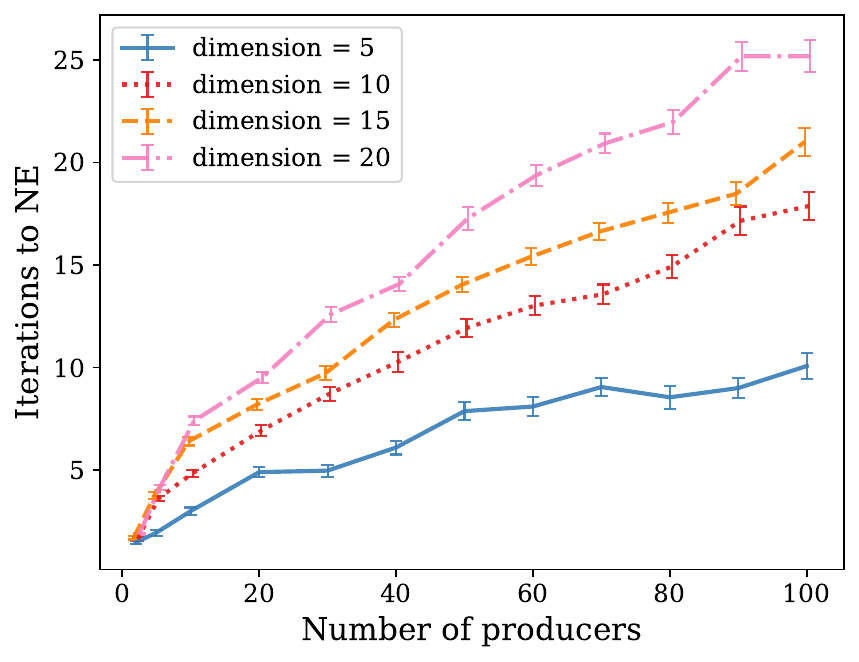}\label{fig:convlinear-AmazonMusic}
}
\hfill
\subfloat[\centering Softmax serving AmazonMusic]{
    \includegraphics[width=0.23\textwidth]{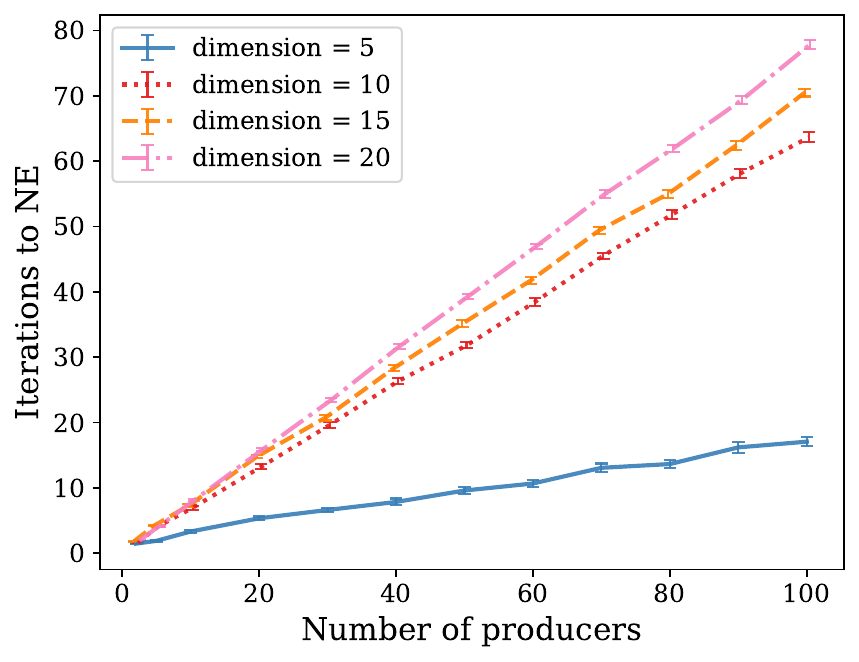}\label{fig:convsoftmax-AmazonMusic}
} 
\hfill
\subfloat[\centering Linear serving RentTheRunway]{
    \includegraphics[width=0.23\textwidth]{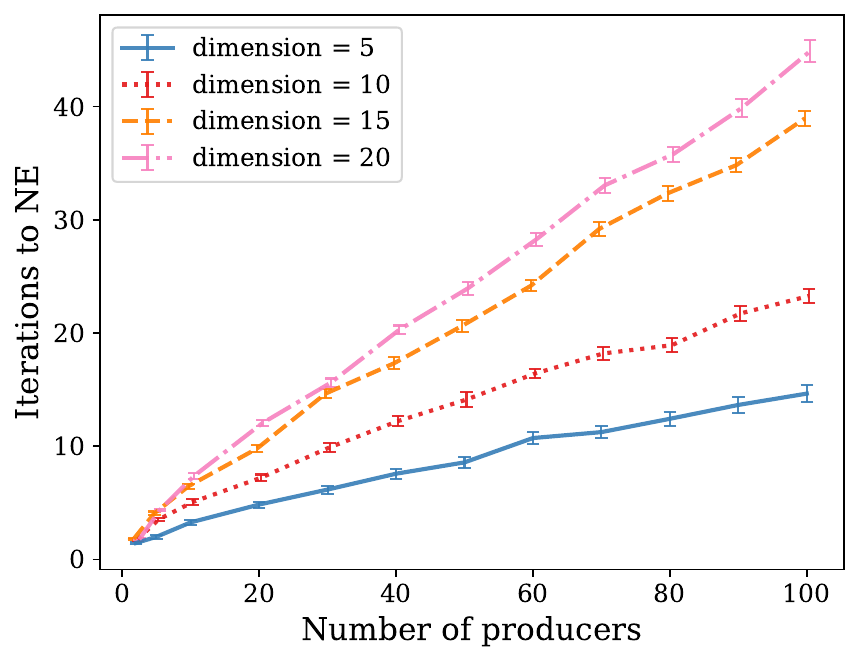}
    \label{fig:convlinear-RentTheRunway}
}
\hfill
\subfloat[\centering Softmax serving RentTheRunway]{
    \includegraphics[width=0.23\textwidth]{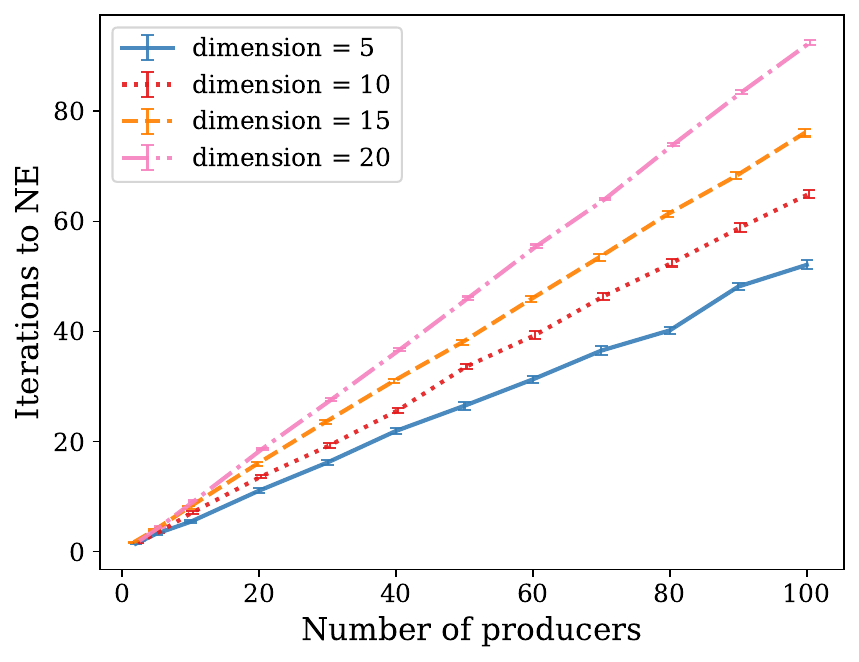}
    \label{fig:convsoftmax-RentTheRunway}
} 
\caption{Number of iterations in Algorithm~\ref{alg:bestrep_dynamics} until convergence to a NE on the AmazonMusic and RentTheRunway datasets. The different curves represent different embedding dimensions in the game $d \in \{5,10,15,20\}$, error bars represent standard error over $40$ runs.}\label{fig:num_iters_convergence_amzn-rtr}
\end{figure}

\subsection{Producer distribution at Nash equilibrium} \label{app:pd-NE}
Figures \ref{fig:full-amznmusic-udpd} and \ref{fig:full-rentrunway-udpd} provide plots for the producer distribution at Nash equilibrium on the AmazonMusic and RentTheRunway dataset respectively. We note that our insights on how the softmax temperature affects specialization at equilibrium also arise in these two additional datasets: namely, lower temperatures lead to higher degrees of specialization.

\begin{figure}[!h]
\centering
\captionsetup[subfigure]{justification=centering}
\subfloat[Softmax $\tau = 100$]{
    \includegraphics[width=0.25\textwidth]{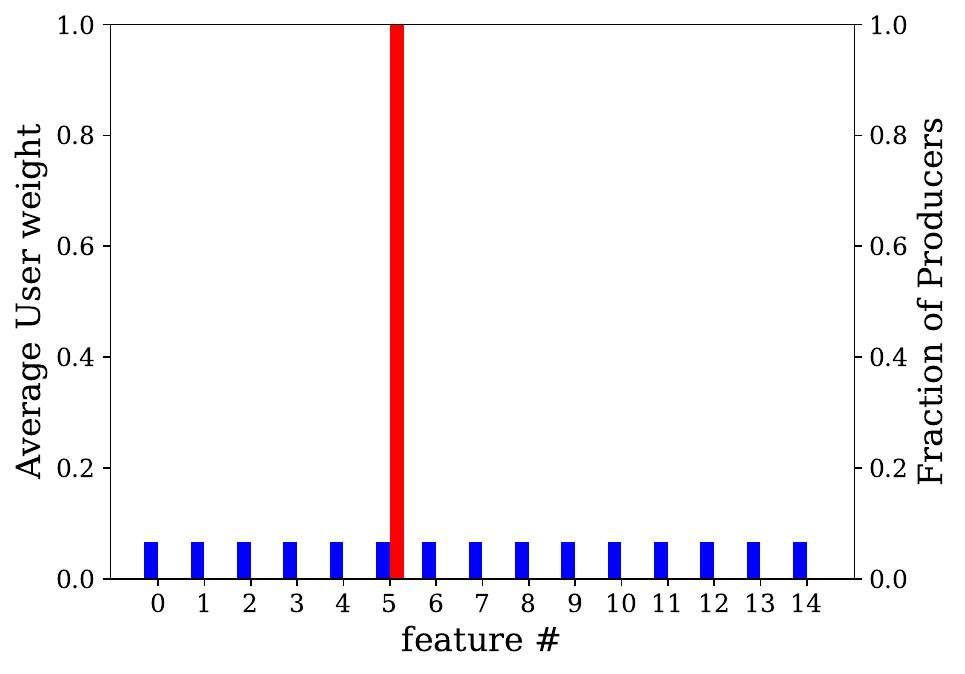}
    \label{fig:amznmusic-udpd-tau100}
}
\hspace{0.05\textwidth}   
\subfloat[Softmax $\tau = 10$]{
    \includegraphics[width=0.25\textwidth]{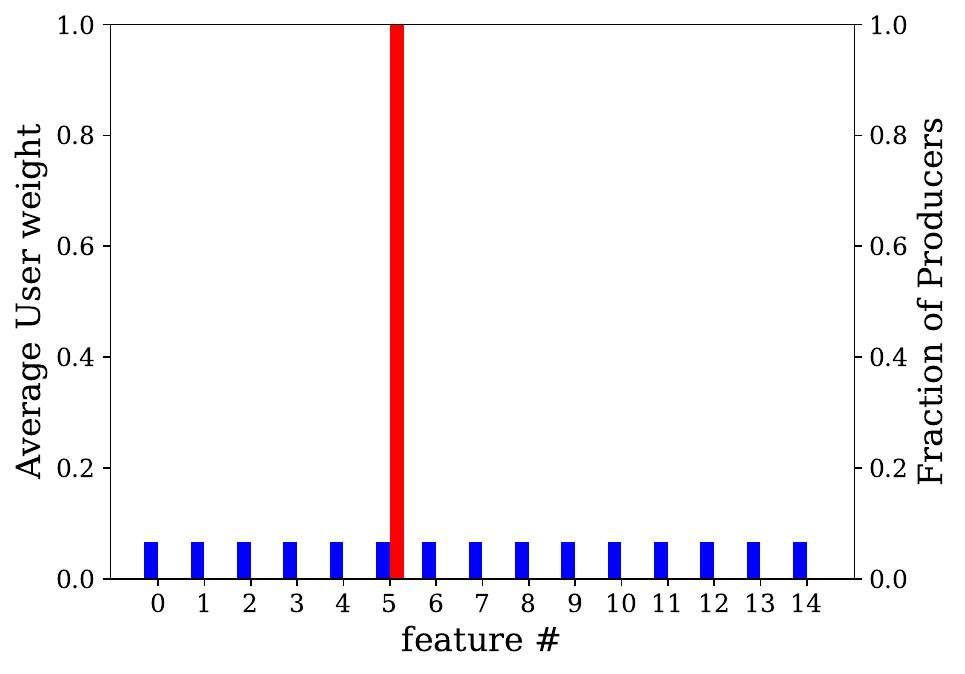}
}
\hspace{0.05\textwidth}   
\subfloat[Softmax $\tau = 1$]{
    \includegraphics[width=0.25\textwidth]{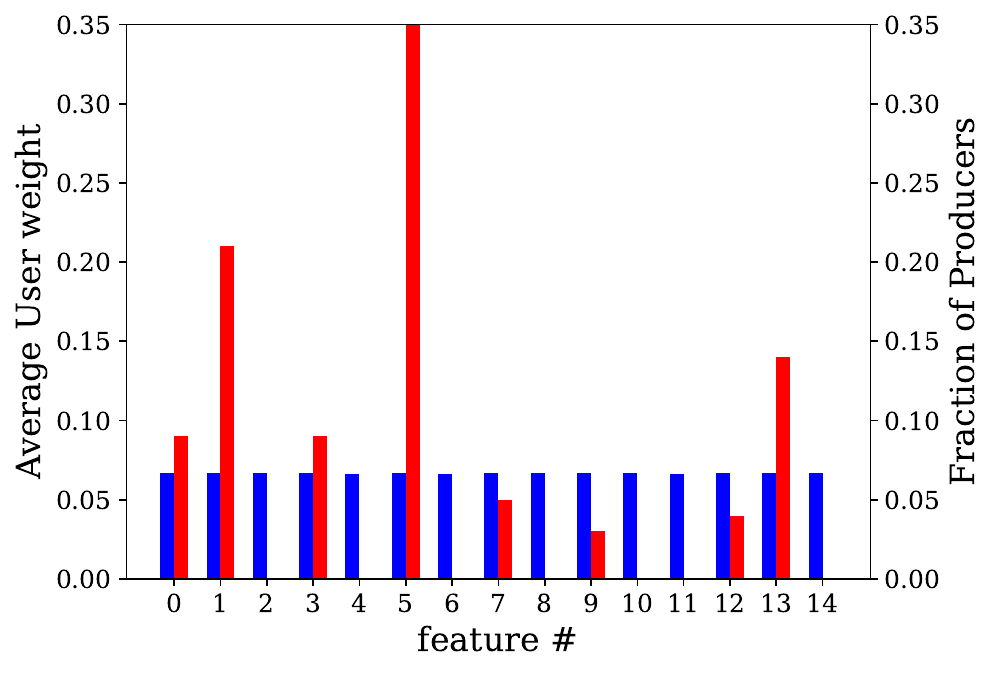}
}
\\
\subfloat[Softmax $\tau = 0.1$]{
    \includegraphics[width=0.25\textwidth]{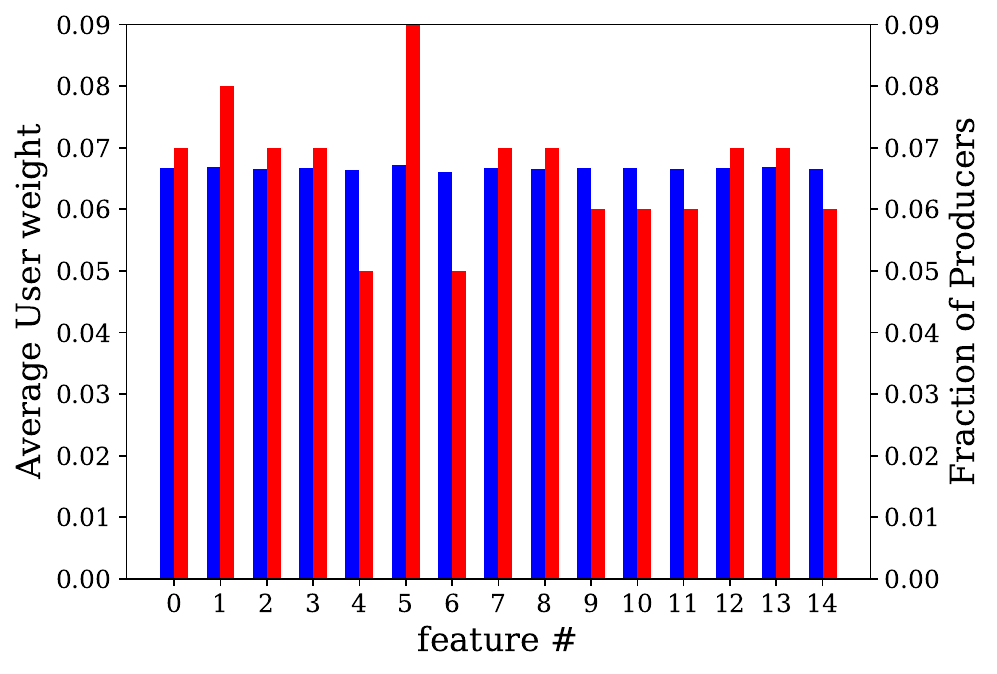}
}
\hspace{0.05\textwidth}   
\subfloat[Softmax $\tau = 0.01$]{
    \includegraphics[width=0.25\textwidth]{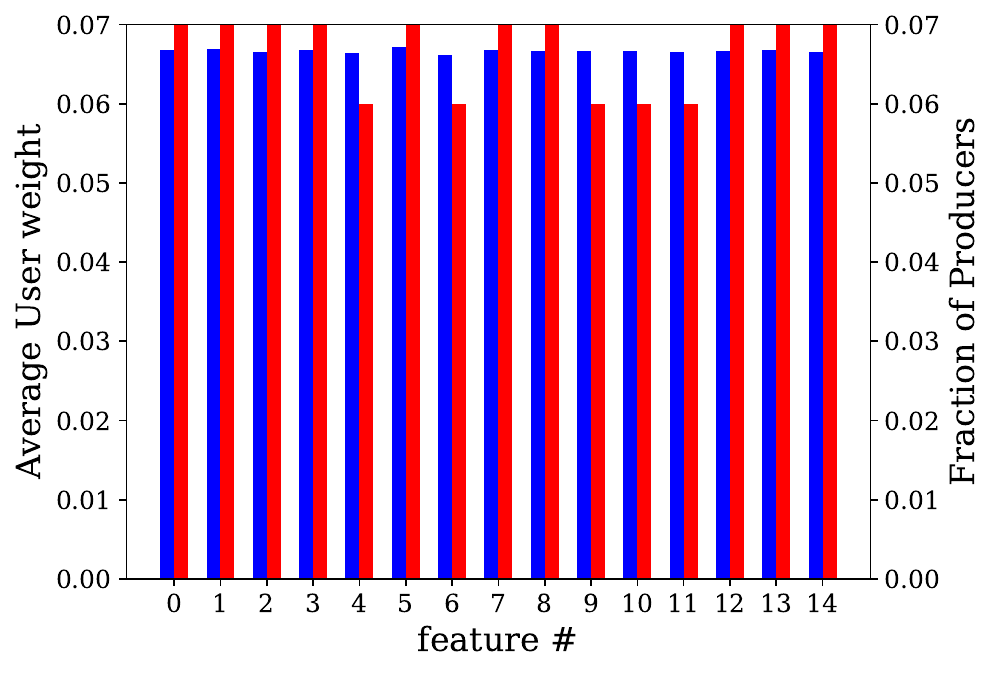}
    \label{fig:amznmusic-udpd-tau001}
}
\hspace{0.05\textwidth}   
\subfloat[the linear-proportional serving rule]{
    \includegraphics[width=0.25\textwidth]{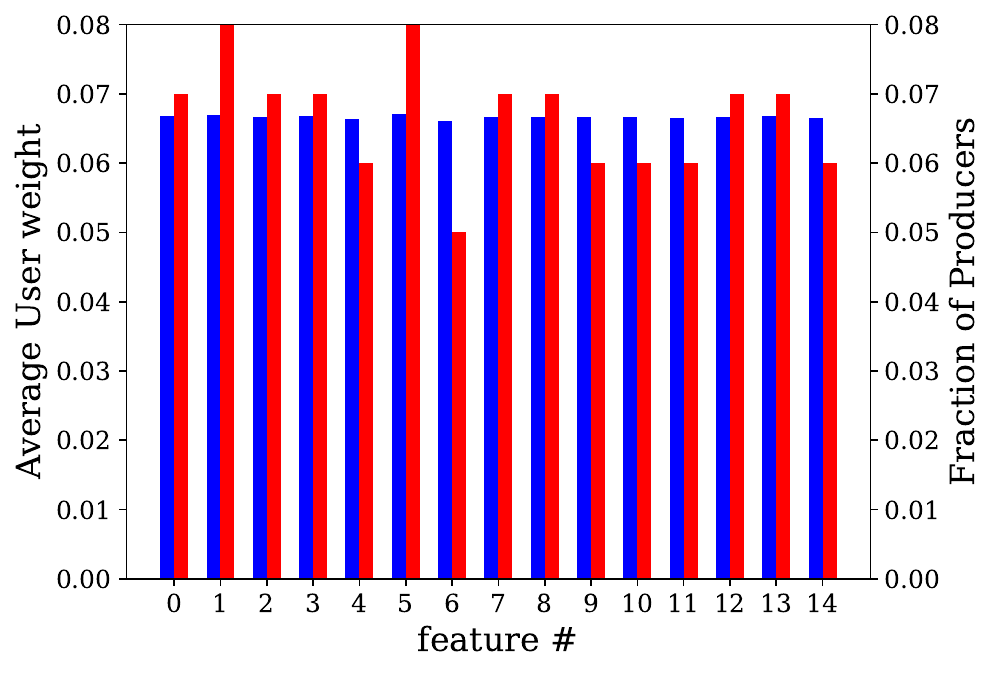}
    \label{fig:amznmusic-udpd-linear}
}
\caption{Average user weight on each feature (blue, left bar) and fraction of producers going for each feature (red, right bar) $n = 100$ producers, embedding dimension $d = 15$. Lower softmax temperature leads to more producer specialization. User embeddings obtained from NMF on the AmazonMusic dataset. \label{fig:full-amznmusic-udpd}}
\end{figure}

\begin{figure}[!h]
\centering
\captionsetup[subfigure]{justification=centering}
\subfloat[Softmax $\tau = 100$]{
    \includegraphics[width=0.25\textwidth]{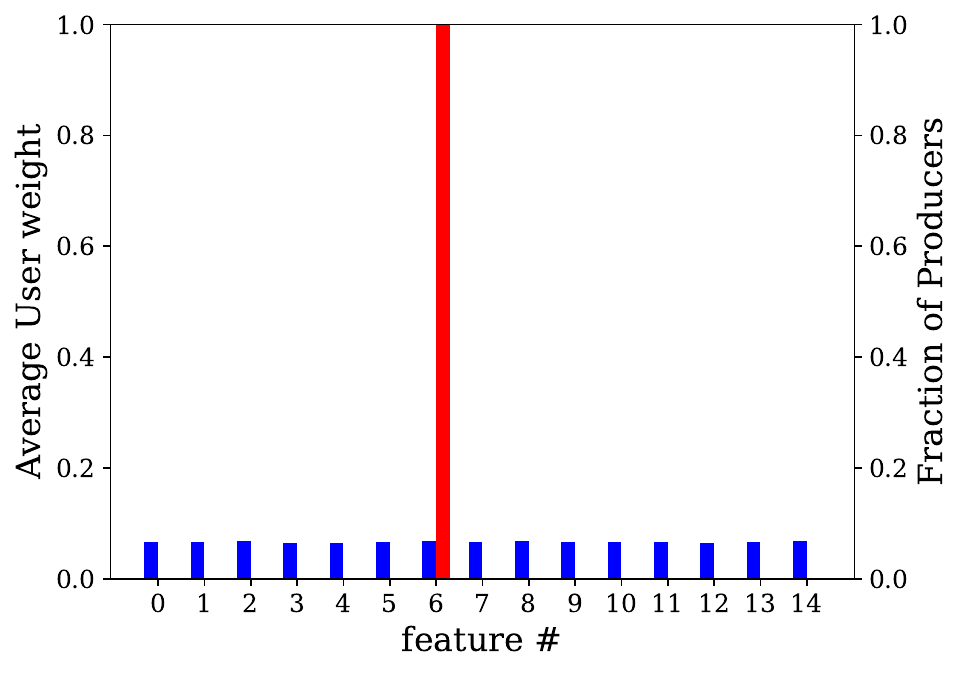}
    \label{fig:rentrunway-udpd-tau100}
}
\hspace{0.05\textwidth}   
\subfloat[Softmax $\tau = 10$]{
    \includegraphics[width=0.25\textwidth]{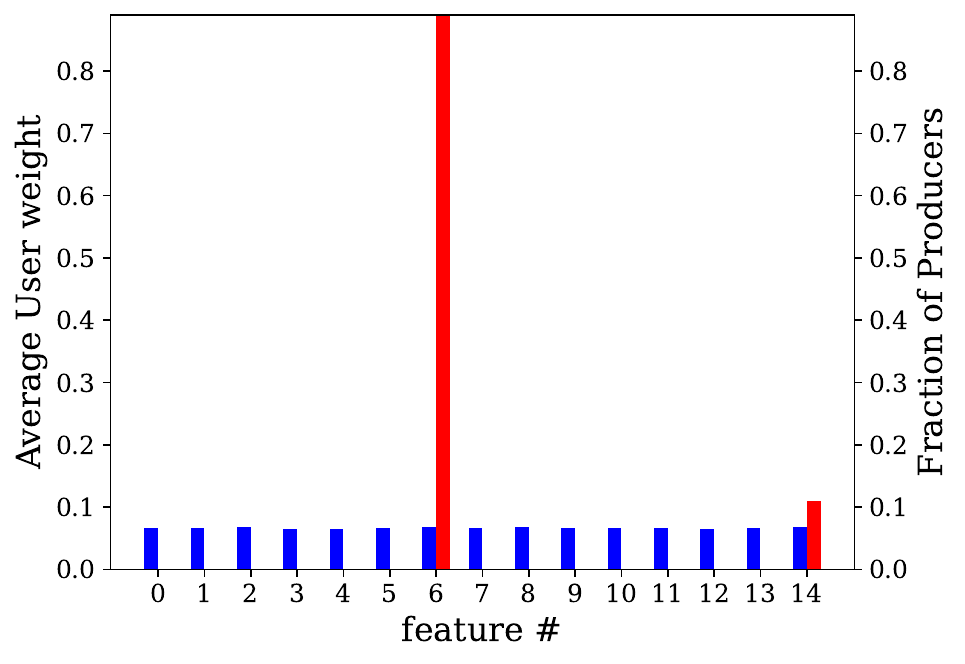}
}
\hspace{0.05\textwidth}   
\subfloat[Softmax $\tau = 1$]{
    \includegraphics[width=0.25\textwidth]{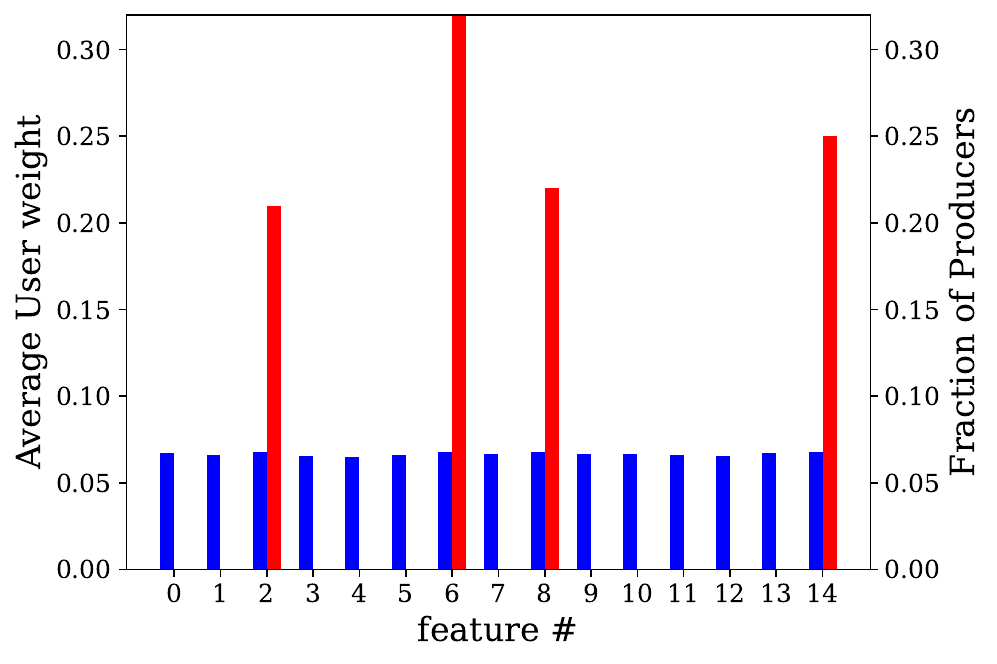}
}
\\
\subfloat[Softmax $\tau = 0.1$]{
    \includegraphics[width=0.25\textwidth]{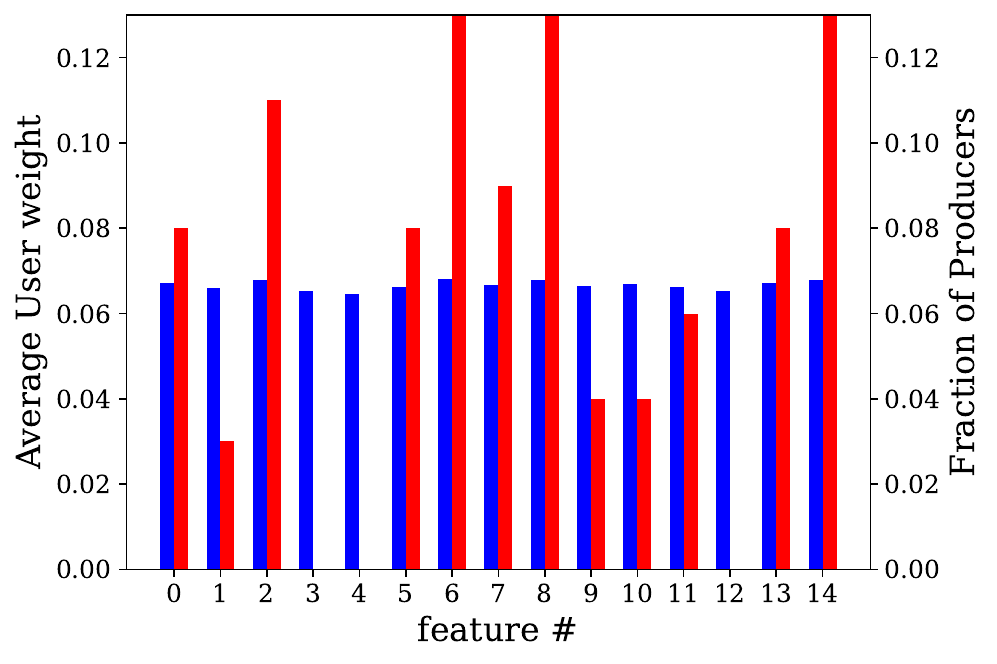}
}
\hspace{0.05\textwidth}   
\subfloat[Softmax $\tau = 0.01$]{
    \includegraphics[width=0.25\textwidth]{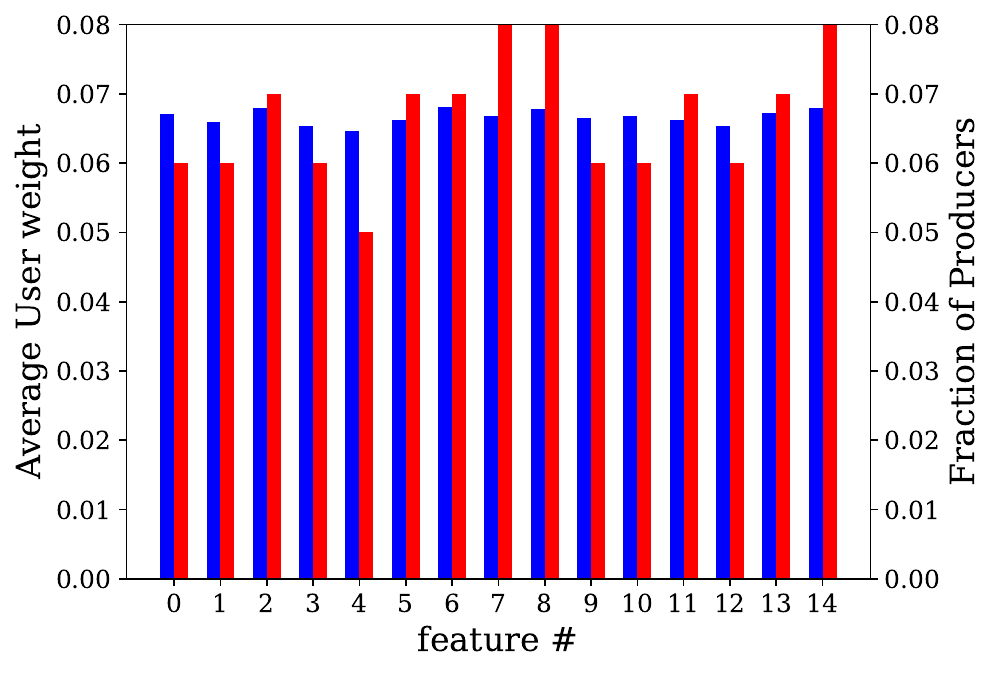}
    \label{fig:rentrunway-udpd-tau001}
}
\hspace{0.05\textwidth}   
\subfloat[the linear-proportional serving rule]{
    \includegraphics[width=0.25\textwidth]{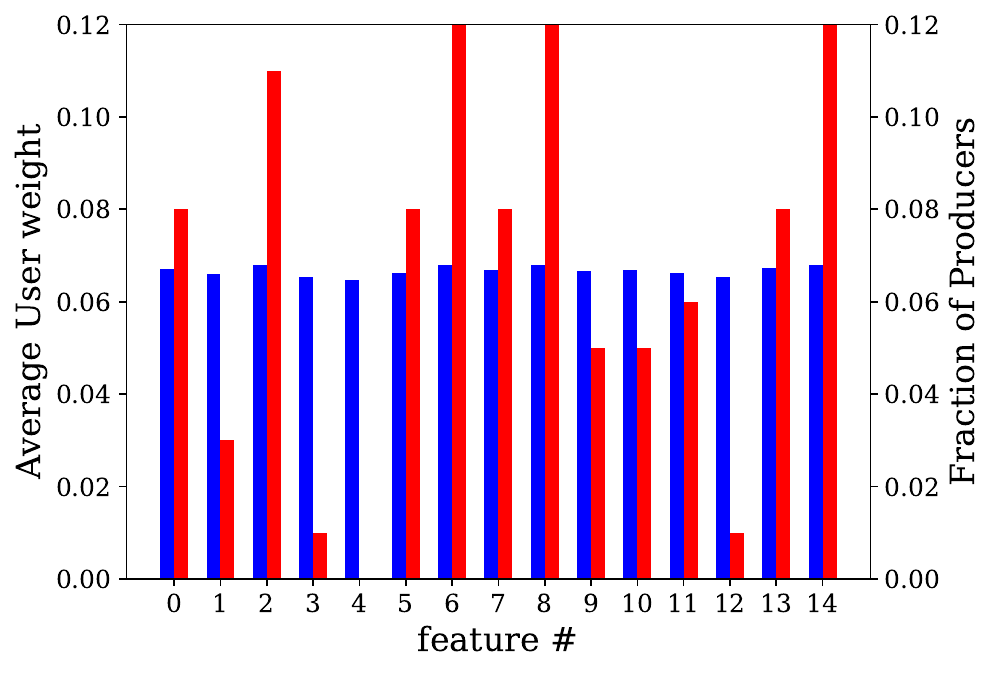}
    \label{fig:rentrunway-udpd-linear}
}
\caption{Average user weight on each feature (blue, left bar) and fraction of producers going for each feature (red, right bar) $n = 100$ producers, embedding dimension $d = 15$. Lower softmax temperature leads to more producer specialization. User embeddings obtained from NMF on the RentTheRunway dataset. \label{fig:full-rentrunway-udpd}}
\end{figure}

\subsection{Average producer utility}\label{app:pu-utils}
In Figure \ref{fig:utility_plot_app-AM-RTR}, we plot the average producer utility with increasing softmax-serving temperature and with varying numbers of producers on the AmazonMusic and RentTheRunway datasets. As in Figure \ref{fig:utility_producers-ml100k}, we observe that the producer utility is decreasing with temperature, and temperature $\tau = 0.01$ (near-hardmax) has the highest utility. This further supports using low temperatures in the softmax content-serving rule.

\begin{figure}[!h]
\centering
\captionsetup[subfigure]{justification=centering}
\subfloat[\centering Producer utility AmazonMusic]{
    \includegraphics[width=0.23\textwidth]{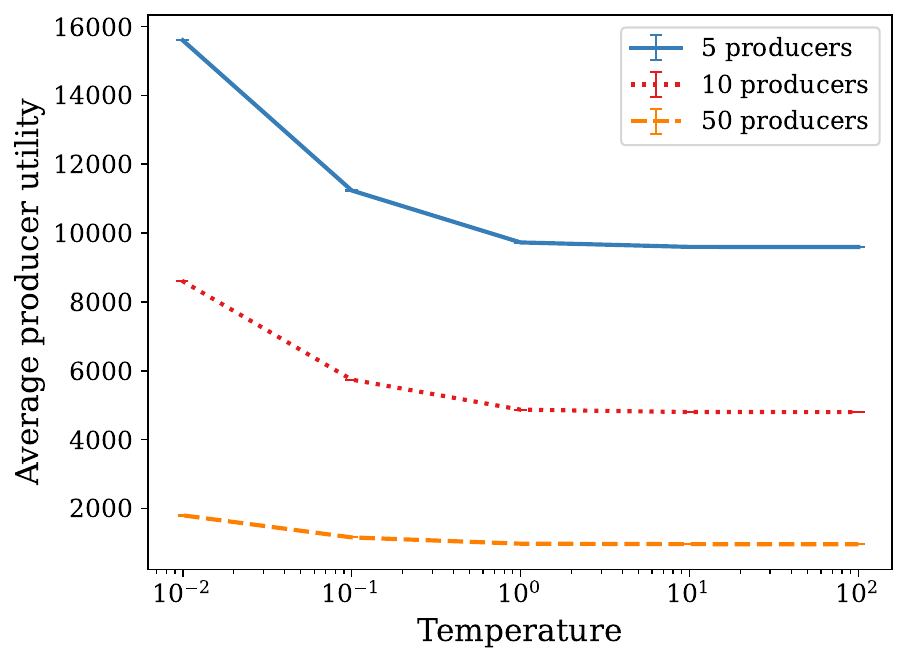}}
\hspace{0.05\textwidth}
\subfloat[\centering Producer utility RentTheRunway]{
    \includegraphics[width=0.23\textwidth]{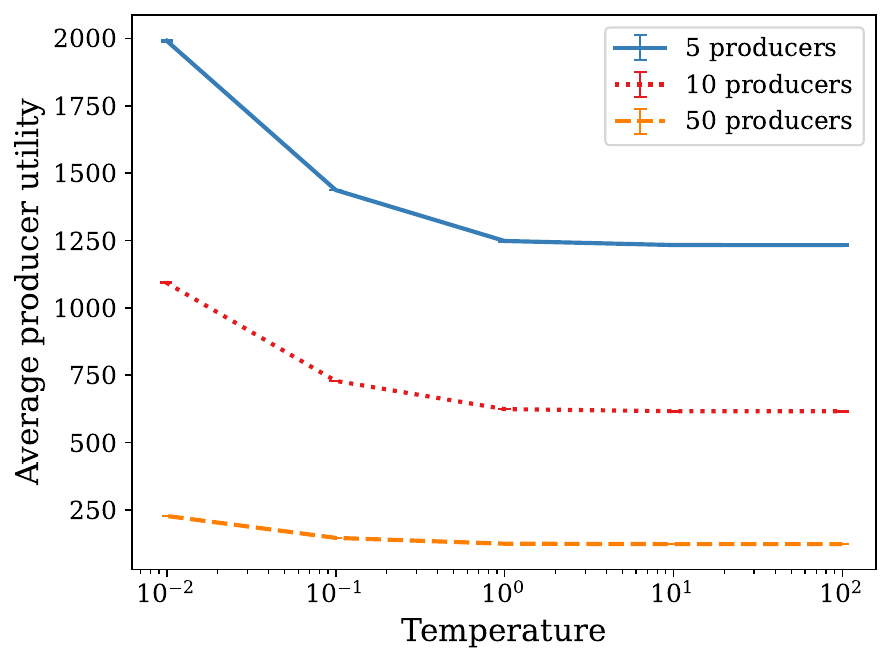}
}
\caption{Average producer utility is decreasing in the softmax temperature, results on the AmazonMusic and RentTheRunway datasets. The different curves represent different number of producers in the game $n \in \{5,10,50\}$, embedding dimension $d= 15$, error bars represent standard error over $5$ seeds} \label{fig:utility_plot_app-AM-RTR}
\end{figure}

In Figure~\ref{fig:utility_plot_app-AM-RTR-linvssm} we compare the average producer utility with the softmax serving rule (across increasing temperatures) v.s with the linear-proportional serving rule. We observe that across both datasets, the lowest softmax temperature we experiment with ($\tau = 0.01$) leads to a greater utility when compared to linear serving. However, linear serving still obtains a competitive utility, greater than that with softmax temperatures $\tau \in \{0.1, 1, 10, 100\}$.

\begin{figure}[!h]
\centering
\captionsetup[subfigure]{justification=centering}
\subfloat[\centering Producer utility AmazonMusic]{
    \includegraphics[width=0.23\textwidth]{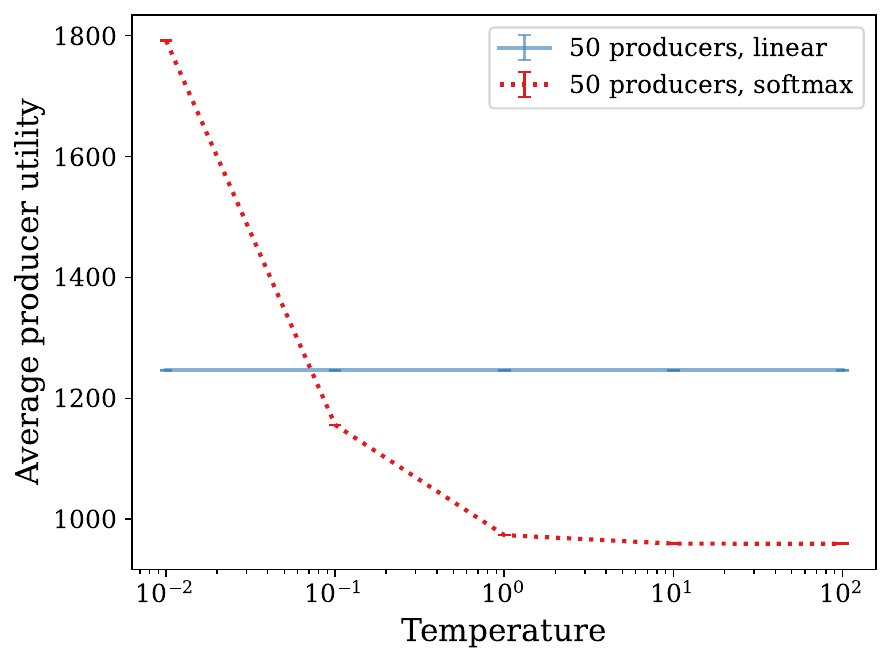}}
\hspace{0.05\textwidth}
\subfloat[\centering Producer utility RentTheRunway]{
    \includegraphics[width=0.23\textwidth]{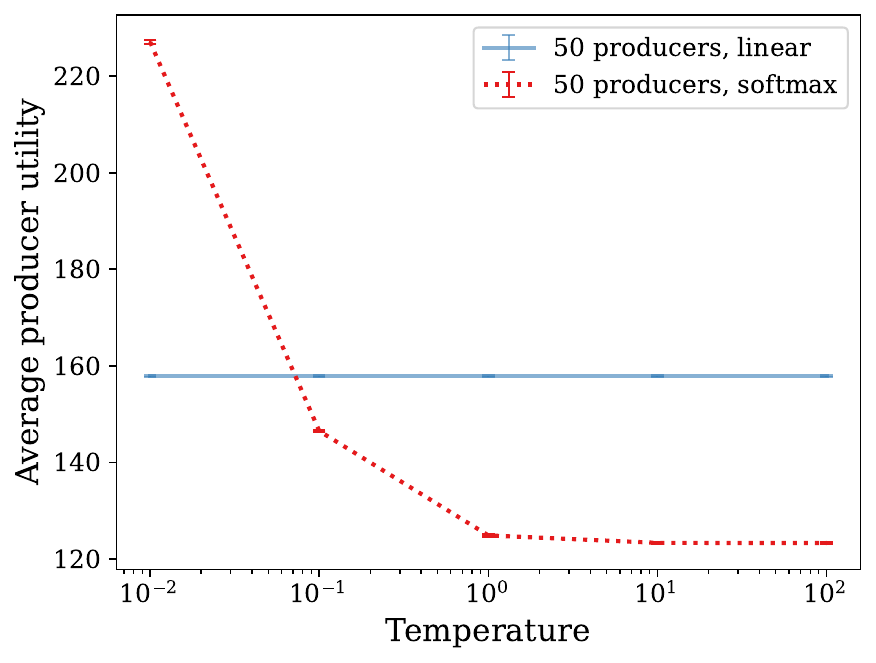}
}
\caption{Average producer utility with linear serving (in blue) v.s with softmax-serving (in red) across increasing temperatures on the AmazonMusic and RentTheRunway datasets. $n=50$ producers, embedding dimension $d=15$, error bars represent standard error over $5$ seeds.}
\label{fig:utility_plot_app-AM-RTR-linvssm}
\end{figure}

%%%%%%%%%%%%%%%%%%%%%%%%%%%%%%%%%%%%%%%%%%%%%%%%%%%%%%%%%%%%
\section{Tables for Nash Equilibrium convergence}
\label{appendix-NE-convtables}
For our engagement game, we consider $12$ different number of producers $(2, 5, 10, 20, \ldots, 100)$ and $4$ different embedding dimensions $(5, 10, 15, 20)$. Each of these $12 \times 4$ games are instantiated with $5$ random seeds for the randomness in the draws of the synthetic data, and the randomness of NMF algorithm on the real datasets. \\
We observe that on each of these $240$ instances of (\#producer, dimension, seed), the softmax content-serving rule with temperatures $\tau \in \{100, 10, 1, 0.1\}$ \emph{always} converges to a unique NE on all the datasets.  With softmax temperature $\tau=0.01$, Algorithm \ref{alg:bestrep_dynamics} still converges to a unique NE in a large number of instances. The linear-proportional serving also converges to a unique NE almost always.

\begin{table}[!h]
\centering
\begin{tabular}{lrrrrrr}
\hline
\multirow{2}{*}{Dataset $\backslash$ Serving Rule} & Linear & \multicolumn{5}{c}{Softmax} \\ 
                         &        & $\tau = 100$ & $\tau = 10$ & $\tau = 1$ & $\tau = 0.1$ & $\tau = 0.01$ \\ 
\hline
Uniform                  & 239    & 240          & 240         & 240         & 240           & 232            \\ 
Skewed-uniform           & 240    & 240          & 240         & 240         & 240           & 240            \\ 
Movielens-100k           & 240    & 240          & 240         & 240         & 240           & 233            \\ 
AmazonMusic              & 240    & 240          & 240         & 240         & 240           & 236            \\ 
RentTheRunway            & 240    & 240          & 240         & 240         & 240           & 237            \\ 
\hline
\end{tabular}
\captionof{table}{Number of instances in which Algorithm~\ref{alg:bestrep_dynamics} converges to a Nash equilibrium \label{table:expconv_table}}
\end{table}

\section{Reproducibility} \label{app:reproducibility}
Here we briefly describe the datasets, embedding generation, and experiments. The code and further details are available in the supplementary.

\paragraph{Datasets:} We import the Movielens-100k dataset from the \code{scikit-surprise} package. For the AmazonMusic ratings we use the ``ratings only'' \code{Digital\_Music.csv} from \url{https://nijianmo.github.io/amazon/index.html}, this dataset has approximately $1.5$ million ratings, $840k$ unique users and $450k$ unique items.  For RentTheRunway we use  \url{https://cseweb.ucsd.edu/~jmcauley/datasets.html#clothing_fit}, this dataset has around $190k$ ratings, $100k$ unique users and $5.8k$ unique items. 
\paragraph{Embedding generation:}
We consider the following embedding seeds $\{13, 17, 19, 23, 29\}$ for randomness in the embedding generation. These $5$ embeddings seeds are used in the random draws for the synthetic uniform and skewed embeddings, and for the randomness in the Non-negative matrix factorization embeddings for the Movielens-100k, AmazonMusic and RentTheRunway datasets.
Note that the real data embeddings are generated using the NMF implementation in \code{scikit-surprise} where we pick $4$ different factors $d \in \{5, 10, 15, 20\}$. 
The synthetic and real data embeddings take 2.6 CPU hours in total to generate and are saved for further use in the experiments.

Note that to parallelize workloads for embedding generation and for the following experiments we use Slurm job arrays.

\paragraph{Experiment 1 : Convergence of Algorithm \ref{alg:bestrep_dynamics}}
For our engagement game, we consider $12$ different number of producers $n \in \{2, 5, 10, 20, 30, 40, 50, 60, 70, 80, 90, 100\}$ and $4$ different embedding dimensions $ d \in \{5, 10, 15, 20\}$. Each of these $12 \times 4$ engagement games are instantiated with the $5$ embedding seeds described above which determine the user embeddings for a dataset.
We then run Algorithm \ref{alg:bestrep_dynamics} $40$ times with the maximum number of iterations $N_{max}$ set to 500. In each of the $40$ trials of Algorithm \ref{alg:bestrep_dynamics} the producer strategies are randomly initialized using the sequential seeds $\{1, 2, \ldots 40\}$. We then plot the mean and standard error across these in Figure \ref{fig:num_iters_convergence},\ref{fig:num_iters_convergence-synthetic} and \ref{fig:num_iters_convergence_amzn-rtr}. All plots use a softmax temperature of $\tau = 1$ for simplicity of exposition.

This experiment on the Uniform, skewed and Movielens-100k datasets take approximately $50, 74$ and $3$ total CPU hours respectively. With the larger scale AmazonMusic and RentTheRunway it takes $\sim$ 4.5k and 1k total CPU hours.  

\textbf{Experiment 2 : Producer distribution and utility \\\\} 
\textbf{Producer distribution:}
Here we fix the embedding seed to $17$, and plot the producer distribution for the linear serving rule and for the softmax-serving rule after Algorithm \ref{alg:bestrep_dynamics} terminates. 
Note that here we don't report mean producer distribution across embedding seeds as this can hide lack of specialization, in the following manner: for a given embedding, only a few features are targeted by the producers, like in Figure~\ref{fig:uniform-udpd-temp10-full} all producers go for index $\#2$. However, the specific index of the feature could change as the embedding seed changes. In that case, when averaging over the seeds, we may be under the impression that specialization occurs, when it does not. 

\textbf{Producer utility:}
Here for each of the $12 \times 4$ instances of \#producers $n$ and dimension $d$, we instantiate the engagement game with linear-serving and $5$ different softmax-serving temperatures $\tau \in \{0.01, 0.1, 1, 10, 100\}$. We measure the average producer and user utility at the unique Nash equilibrium and report its mean across the $5$ embedding seeds. Average utilities are measured by taking the average across converging runs, and non-converging seeds are dropped. 

This experiment on the Uniform, skewed and Movielens-100k datasets take approximately $4.4, 7$ and $0.5$ total CPU hours respectively. With the larger scale datasets, AmazonMusic and RentTheRunway, it takes $\sim$ 530 and 105 total CPU hours.

\end{document}